\documentclass{theoretics}

\title{Fully Dynamic Connectivity in \texorpdfstring{$O(\log n(\log \log n)^2)$}{O(log n(log log n)2)} Amortized Expected Time}

\ThCSauthor[affil1]{Shang-En Huang}{sehuang@umich.edu}[0000-0002-9799-0981]
\ThCSauthor[affil1]{Dawei Huang}{dwhuang9@gmail.com}[0000-0002-4496-2354]
\ThCSauthor[affil2]{Tsvi Kopelowitz}{kopelot@gmail.com}[0000-0002-3525-8314]
\ThCSauthor[affil1]{Seth Pettie}{seth@pettie.net}[0000-0002-0495-3904]
\ThCSauthor[affil3]{Mikkel Thorup}{mikkel2thorup@gmail.com}[0000-0001-5237-1709]

\ThCSaffil[affil1]{Computer Science and Engineering, University of Michigan}
\ThCSaffil[affil2]{Department of Computer Science, Bar-Ilan University}
\ThCSaffil[affil3]{Department of Computer Science, University of Copenhagen}

\ThCSyear{2023}
\ThCSarticlenum{6}
\ThCSreceived{June 3, 2022}
\ThCSrevised{Oct 14, 2022}
\ThCSaccepted{Mar 15, 2023}
\ThCSpublished{April 28, 2023}
\ThCSkeywords{Dynamic graph, amortized analysis, connectivity}
\ThCSdoi{}
\ThCSshortnames{S.-E. Huang, D. Huang, T. Kopelowitz, S. Pettie and M. Thorup}
\ThCSshorttitle{Fully Dynamic Connectivity, Amortized Expected Time}
\ThCSthanks{This paper is the successor of two extended abstracts published in
STOC 2000 (Thorup~\cite{Thorup2000}) and SODA 2017 (Huang, Huang, Kopelowitz,
and Pettie~\cite{HuangHKP17}). Supported by NSF grants CCF-1217338, CNS-1318294,
CCF-1514383, CCF-1637546, and CCF-1815316. Mikkel Thorup's research is supported by
his Investigator Grant 16582, Basic Algorithms Research Copenhagen (BARC), from the
VILLUM Foundation.}

\usepackage{verbatim}
\usepackage{soul}
\usepackage{comment}

\usepackage{tikz}
\usepackage{ifthen}
\usepackage{subcaption}
\usetikzlibrary{calc}
\usetikzlibrary{decorations.pathmorphing}
\usetikzlibrary{shapes.geometric}
\usetikzlibrary{backgrounds}
\usetikzlibrary{intersections}
\tikzset{snake it/.style={decorate, decoration=snake}}

\newcommand{\paragraphb}[1]{\paragraph{#1}}

\ThCSnewtheoita{invariant}

\newcommand{\ignore}[1]{}

\newcommand{\floor}[1]{\lfloor #1 \rfloor}
\newcommand{\f}[2]{\frac{#1}{#2}}
\newcommand{\fr}[2]{\mbox{$\frac{#1}{#2}$}}
\newcommand{\paren}[1]{\left( #1 \right)}

\newcommand{\bottom}{\perp}

\newcommand{\scut}[2]{\ensuremath{#1 \leftrightharpoons #2}}
\newcommand{\Endpoint}[2]{\ensuremath{\langle #1,#2\rangle}}
\newcommand{\T}{\mathcal{T}}
\newcommand{\W}{\mathcal{W}}
\renewcommand{\H}{\ensuremath{\mathcal{H}}}
\renewcommand{\P}{\ensuremath{\mathcal{P}}}

\newcommand{\F}{\mathcal{F}}
\newcommand{\DF}{{\mathfrak{D}_\F}}
\renewcommand{\DH}{{\mathfrak{D}_\H}}
\renewcommand{\L}{\mathcal{L}}
\newcommand{\dmax}{{d_{\mathit{max}}}}

\newcommand{\isConnected}{{\texttt{Conn?}}}
\newcommand{\Insert}{{\texttt{Insert}}}
\newcommand{\Delete}{{\texttt{Delete}}}

\newcommand{\ang}[1]{\left< #1 \right>}

\newcommand{\const}[1]{\operatorname{#1}}

\DeclareMathOperator{\LSBIndex}{\mathrm{LSBIndex}}
\DeclareMathOperator{\ShortCuts}{\textsc{Cover}}

\newcommand{\Usc}{\textsc{Up}}
\newcommand{\Uptr}{\textsc{UpIdx}}
\newcommand{\Dsc}{\textsc{Down}}
\newcommand{\Dptr}{\textsc{DownIdx}}
\newcommand{\Occ}{\textsc{Occ}}

\newcommand{\Patrascu}{P\v{a}tra\c{s}cu}

\def\compactify{\itemsep=0pt \topsep=0pt \partopsep=0pt \parsep=0pt}

\makeatletter
\let\oldcodebox\codebox
\def\codebox{\oldcodebox
\hspace*{2em}\=\hspace*{2em}\=\hspace*{2em}\=\hspace*{2em}\=\hspace*{24em}\=\=\=\=\=\=\=\=\kill}
\makeatother

\begin{document}
\maketitle

\begin{abstract}
Dynamic connectivity is one of the most fundamental problems in dynamic graph algorithms.
We present a randomized Las Vegas dynamic connectivity data structure with
$O(\log n(\log\log n)^2)$ amortized expected update time
and $O(\log n/\log\log\log n)$ worst case query time,
which comes very close to the cell probe lower bounds of
\Patrascu{} and Demaine (2006) and \Patrascu{} and Thorup (2011).
\end{abstract}


\section{Introduction}

The \emph{dynamic connectivity} problem is one of the most fundamental problems in dynamic graph algorithms.
The goal is to support the following three operations on an undirected graph $G=(V,E)$ 
with $n=|V|$ vertices and $m=|E|$ edges,
where $E$ is initially empty.
\begin{itemize}
    \item $\Insert(u, v)$: Set $E \leftarrow E \cup \{\{u,v\}\}$.
    \item $\Delete(u, v)$: Set $E \leftarrow E - \{\{u,v\}\}$.
    \item $\isConnected(u, v)$: Return {\tt{true}} if and only if $u$ and $v$ are in the same connected component in $G$.
\end{itemize}

Dynamic connectivity has been studied in numerous models, 
under both worst case and amortized notions of time,
with deterministic, randomized Las Vegas, and randomized Monte Carlo guarantees, and with both \emph{public} and \emph{private} 
witnesses of connectivity.  Las Vegas algorithms always answer $\isConnected$ queries correctly but their running time is a random variable.  In contrast, the running time of a Monte Carlo algorithm
is guaranteed deterministically, but it only answers $\isConnected$ queries correctly with high probability.  All known dynamic connectivity algorithms maintain a spanning forest $F$ of $G$ as a sparse certificate of connectivity.  If $F$ is \emph{public} then the sequence of 
$\Insert$ and $\Delete$ operations may depend on $F$, and may 
therefore depend on random bits generated earlier by the data structure.
When $F$ is \emph{private} the $\Insert$/$\Delete$ sequence
is selected obliviously.

In this paper we prove near-optimal bounds on the \emph{amortized} complexity of dynamic connectivity in the \emph{Las Vegas} randomzed model, with a \emph{public} connectivity witness.

\begin{theorem}\label{main-result}
There exists a Las Vegas randomized dynamic connectivity data structure
that supports insertions and deletions of edges in amortized expected
$O(\log n(\log \log n)^2)$ time,
and answers connectivity queries in worst case 
$O(\log n / \log \log\log n)$ time.
The time bounds hold even if the adversary is aware of
the internal state of the data structure.
In particular, the data structure maintains 
a public spanning forest as a connectivity witness.
\end{theorem}

\subsection{A Brief History of Dynamic Connectivity Data Structures}

\paragraph{Worst Case Time.}
Frederickson~\cite{Frederickson85} developed a dynamic connectivity structure in the strictest
model---deterministic worst case time---with $O(\sqrt{m})$ update time and $O(1)$ query time.
Eppstein, Galil, Italiano, and Nissenzweig~\cite{EppsteinGIN97} showed that the update times
for many dynamic graph algorithms could be made to depend on $n$ rather than $m$, provided they maintain
an $O(n)$-edge witness of the property being maintained, e.g., a spanning forest in the case of dynamic connectivity.
Together with~\cite{Frederickson85}, Eppstein et al.'s reduction implied an $O(\sqrt{n})$ update 
time for dynamic connectivity,
a bound which stood for many years.  
Kejlberg, Kopelowitz, Pettie, and Thorup~\cite{KejlbergKPT16}
simplified Frederickson's data structure, and improved the update time of~\cite{Frederickson85,EppsteinGIN97}
to $O\paren{\sqrt{\frac{n(\log\log n)^2}{\log n}}}$. Recently Chuzhoy, Gao, Li, Nanongkai, Peng, and Saranurak~\cite{ChuzhoyGLNPS20} improved the worst case update time to $n^{o(1)}$.

Kapron, King, and Mountjoy~\cite{KapronKM13}
gave a Monte Carlo randomized structure with update 
time $O(c\log^5 n)$ and one-sided error\footnote{An error occurs from reporting that
two vertices are disconnected when they are actually connected.} probability $n^{-c}$.
Their data structure maintains a \emph{private} connectivity witness, i.e., it keeps a spanning tree, but 
the adversary controlling $\Insert$ and $\Delete$ operations 
does not have access to the spanning tree.
The update time was later improved to $O(c\log^4 n)$ independently by
Gibb et al.~\cite{GibbKKT15} and Wang~\cite{Wang15}, 
and Wang further reduced the time for $\Insert$ to $O(c\log^3 n)$.
Nanongkai, Saranurak, and Wulff-Nilsen~\cite{NanongkaiSW17}
discovered a Las Vegas randomized structure with $n^{o(1)}$ 
update time that maintains a public connectivity witness.  
\nocite{Wulff-Nilsen17,NanongkaiS17}
This data structure was recently derandomized~\cite{ChuzhoyGLNPS20},
leading to a deterministic $n^{o(1)}$ dynamic 
connectivity algorithm maintaining a public witness.

\paragraph{Amortized Time.}
By allowing amortization in the running time, dynamic connectivity can be solved even faster.
Henzinger and King~\cite{Henzinger1999} proved that with Las Vegas randomization,
dynamic connectivity could be solved with amortized expected $O(\log^3 n)$ update time.  This was subsequently improved
to amortized expected $O(\log^2 n)$ update time by Henzinger and Thorup~\cite{Henzinger97} using a more sophisticated sampling routine.
Holm, de Lichtenberg, and Thorup~\cite{Holm2001} discovered a \emph{deterministic} data structure with amortized $O(\log^2 n)$ update time.  Thorup~\cite{Thorup2000} later improved the \emph{space} of~\cite{Holm2001}
from $O(m+n\log n)$ to optimal $O(m)$, and Wulff-Nilsen~\cite{Wulff-Nilsen13} further improved~\cite{Holm2001,Thorup2000} to have
amortized $O(\log^2 n/\log\log n)$ update time using $O(m)$ space.

At \emph{STOC 2000}, Thorup~\cite{Thorup2000} presented
a Las Vegas randomized data structure with
amortized expected $O(\log n(\log\log n)^3)$ update time and worst case $O(\log n/\log\log\log n)$ query time.
At \emph{SODA 2017}, Huang, Huang, Kopelowitz, and Pettie~\cite{HuangHKP17}
improved the update time of \cite{Thorup2000} to $O(\log n(\log\log n)^2)$,
and substantiated several claims that were only sketched in~\cite{Thorup2000}.
The data structures presented in~\cite{Thorup2000,HuangHKP17}
are especially notable in light of the lower bounds of \Patrascu{} and Demaine~\cite{PatrascuD06}
and \Patrascu{} and Thorup~\cite{PatrascuT11}.  The first shows that any (amortized or randomized)
dynamic connectivity structure with $O(t(n)\log n)$ update time,
$t(n)=\Omega(1)$,
requires $\Omega(\log n/\log t(n))$ query time.  In particular, the maximum of update and query time is $\Omega(\log n)$.
The second shows that any dynamic connectivity structure with $o(\log n)$ update time requires $n^{1-o(1)}$ query time.
Thus, any data structure with $O(\log n(\log\log n)^2)$ update time must have $\Omega(\log n/\log\log\log n)$ query time,
and for \emph{any} reasonable query time, we cannot improve our update time by more than a $(\log\log n)^2$ factor.
On certain restricted classes of inputs, e.g., trees~\cite{SleatorT83} and planar graphs~\cite{EppsteinITTWY92},
both updates and queries can be supported in $O(\log n)$ worst case time.

\paragraphb{Contribution.} This paper should be considered the successor and full version of both the \emph{STOC 2000} and the \emph{SODA 2017} extended abstracts \cite{Thorup2000,HuangHKP17}, improving the complexity of dynamic connectivity from amortized $O(\log^2 n/\log\log n)$ update time~\cite{Wulff-Nilsen13} to the near-optimal amortized expected $O(\log n(\log\log n)^2)$ update time.


\paragraphb{Organization of the Paper.}
In Section~\ref{sect:prelim} we review several fundamental concepts of dynamic connectivity algorithms.
Section~\ref{section:overview-of-the-data-structure} gives a detailed overview of the algorithm, and lists the primitive operations from the data structure that implements the algorithm.
In Section~\ref{section:main-modules} we describe the main modules of the data structure.
The main modules include: maintaining a binary {\em hierarchical} representation of the graph,
maintaining {\em shortcuts} for efficient navigation around the hierarchy, and maintaining a system
of {\em approximate counters} to support nearly-uniform random sampling.
Each of these modules is explained in detail in
Sections \ref{section:shortcut-infrastructure-on-H}--\ref{section:amortized-analysis-of-shortcut-maintenance}.
Finally, we piece up all the modules from the data structures and describe how primitive operations listed in Section~\ref{section:overview-of-the-data-structure} (Lemma~\ref{lemma-main}) are implemented and analyzed in Section~\ref{section:main-oerations}.
We make some concluding remarks
in Section~\ref{section:conclusion}.

\section{Preliminaries}\label{sect:prelim}

In this section we review some basic concepts and invariants used in prior
dynamic connectivity algorithms~\cite{Henzinger1999,Holm2001,Thorup2000,Wulff-Nilsen13}.

\paragraphb{Witness Edges, Witness Forests, and Replacement Edges.}
A common method for supporting connectivity queries
is to maintain a spanning forest $\F$ of $G$ called the \emph{witness forest}, 
together with a dynamic connectivity data
structure on $\F$ (see Theorem~\ref{lemma:witness-forest} below). 
Each edge in the witness forest is called a {\emph{witness}} edge,
and all other edges are called {\emph{non-witness}} edges.
Deleting a non-witness edge does not change the connectivity.

\paragraph{Update Time and Query Time.}
When describing the dynamic connectivity data structure we only focus on the (amortized) running time of the update operations.  Once this time bound is fixed,
Theorem~\ref{lemma:witness-forest} provides a fast query time, which,
according to \Patrascu{} and Demaine~\cite{PatrascuD06}, cannot be unilaterally improved.

\begin{theorem}[Henzinger and King~\cite{Henzinger1999}]\label{lemma:witness-forest}
For any function $t(n) =\Omega(1)$, there exists a dynamic connectivity data structure
for forests with $O(t(n)\log n)$ update time and $O(\log n/\log t(n))$ query time.
\end{theorem}

\begin{proof}[Proof sketch]
Maintain an Euler tour of each tree in the witness forest and a balanced $t(n)$-ary rooted tree over the Euler tour elements.
The height of each rooted tree is $O(\log_{t(n)} n)$.  A witness edge insertion/deletion imposes $O(1)$ changes
to the Euler tour, which necessitates $O(t(n) \log_{t(n)} n)$ time to update the rooted trees.
A query $\isConnected(u, v)$ finds the representative copies of $u$ and $v$ in the Euler tours,
walks up to their respective roots, and checks if they are equal.
\end{proof}

The difficulty in maintaining a dynamic connectivity data structure is to find a \emph{replacement edge} $e'$ when a witness edge $e\in\F$
is deleted, or determine that no replacement edge exists.  To speed up the search for replacement edges we maintain
Invariant~\ref{invariant:spanning-forest} (below) governing edge {\em depths}.

\paragraphb{Edge Depths.} Each edge $e$ has a depth $d_e\in [1,\dmax]$, where $\dmax = \lfloor \log n \rfloor$.
Let $E_i$ be the set of edges with depth $i$.  All edges are inserted at depth 1 and depths are non-decreasing over time.
Incrementing the depth of an edge is called a \emph{promotion}.
Since we are aiming for $O(\log n(\log \log n)^2)$ amortized time per update,
if the actual time to promote an edge set $S$ is $O(|S|\cdot (\log\log n)^2)$, the amortized time per promotion is zero.
Promotions are performed in order to maintain Invariant~\ref{invariant:spanning-forest}.

\begin{invariant}[The Depth Invariant]\label{invariant:spanning-forest}\label{invariant:weight-property}~
Define $G_i = (V,\bigcup_{j\ge i} E_j)$.
\begin{enumerate}[(1)]\compactify
\item (Spanning Forest Property) $\F$ is a maximum spanning forest of $G$ with respect to the depths.
\item (Weight Property) For each $i\in [1,d_{max}]$, each connected component in the subgraph
$G_i$ contains at most $n/2^{i-1}$ vertices.
\end{enumerate}
\end{invariant}

\paragraphb{Hierarchy of Connected Components.}
Define $\hat{V}_i$ to be in one-to-one correspondence with the connected components of $G_{i+1}$, which are called \emph{$(i+1)$-components}.
If $u\in V$, let $u^i \in \hat{V}_i$ be the unique $(i+1)$-component containing $u$.
Define $\hat G_i=(\hat V_i, \hat E_i)$ to be the {\em multigraph} 
(including parallel edges and loops)
obtained by contracting edges with depth larger than $i$ and discarding 
edges with depth less than $i$,
so $\hat E_i = \{\{u^i,v^i\} \;|\; \{u,v\}\in E_i\}$.
The hierarchy \H{} is composed of the undirected multi-graphs $\hat{G}_\dmax, \hat G_{\dmax-1}, \ldots, \hat G_0$.
An edge $e=\{u, v\}\in E_i$
is said to be \emph{touching} all nodes
$x^j\in \hat V_j$ where either $u^j=x^j$ or $v^j=x^j$.

Let $F_i = E_i \cap \F$ be the set of \emph{$i$-witness edges}; all other edges in $E_i - F_i$ are \emph{$i$-non-witness edges}.
It follows from Invariant~\ref{invariant:spanning-forest} that 
$F_i$ corresponds to a spanning forest of $\hat G_i$,
if one maps the endpoints of $F_i$-edges 
to the contracted vertices of $\hat G_i$.
The {\em weight} $w(u^i)$ of a node $u^i\in \hat V_i$ is the number of vertices in its component:
$w(u^i) = |\{v\in V \;|\; v^i = u^i\}|$.  The data structure explicitly maintains the exact weight of all hierarchy nodes.
The weight property in Invariant~\ref{invariant:spanning-forest} can be restated as $w(u^{i-1})\le n/2^{i-1}$
since $u^{i-1}$ corresponds to the connected component containing $u$ in $G_i$.

\paragraphb{Endpoints.} The \emph{endpoints} of an edge $e=\{u,v\}$ are the pairs $\ang{u,e}$ and $\ang{v,e}$.
At some stage in our algorithm we sample a random endpoint from a set $S$ of endpoints incident to some $V' \subset V$.
An edge $\{u,v\}$ with $u,v\in V'$ could contribute zero, one, or two endpoints to $S$, i.e.,
the endpoints of an edge are often treated independently.
An endpoint $\ang{u,e}$ is said to be \emph{touching} the nodes $u^i\in \hat V_i$ for all $i\in [1, \dmax]$.

\subsection{Computational Model and Lookup Tables}\label{section:lookup-tables}

We assume a standard $O(\log n)$-bit word RAM with the usual repertoire of $AC^0$ instructions.  The data structure uses some non-standard operations on packed sequences of $O(\log\log n)$-bit floating point numbers, which we can simulate by building small lookup tables with size $O(n^\epsilon)$, for some $\epsilon \in (0,1)$.
Since the initial graph is empty, 
the $O(n^\epsilon)$ sized lookup tables can be built on-the-fly, with their 
cost amortized through the operations.\footnote{As long as the number of graph updates is $m \le n$, all edge depths are at
most $\lfloor \log m\rfloor$. Hence, for each $0\le r\le \log \log n$, after the $m=2^{2^r}$-th graph update, the data structure rebuilds the lookup tables of size $O(m^\epsilon)$. The time cost for building the lookup tables during the first $m$ operations is bounded by
$\sum_{i=0}^{\lceil\log\log m\rceil} m^{\frac{1}{2^i} \epsilon} = O(m^\epsilon)$,
which is amortized $o(1)$ per update.}

\subsection{Miscellaneous}

\paragraph{Almost Uniform Sampling.}
We say that an algorithm samples from a set $X$ \emph{$(1+o(1))$-uniformly at random}, if, 
for any element $x\in X$, 
the probability of $x$ being 
returned is $(1+o(1))/|X|$.
(In our algorithm, the $o(1)$ term is roughly 
$1/\log n$, and $|X|$ is at most polynomial in $n$.)

\paragraph{Mergeable Balanced Binary Trees.}
Some parts of our data structure (see Section~\ref{section:local-trees}) 
use off-the-shelf mergeable balanced binary trees.
They should support leaf-insertion and 
leaf-deletion on $T$ in $O(\log|T|)$ time,
and the \emph{merger} of two trees $T_1,T_2$
in $O(\log|T_1|+\log|T_2|)$ time.
The merge operation may create and delete
internal nodes as necessary to ensure balance.
These trees do \emph{not} store elements from a totally ordered set, and do not need to 
support a \emph{search} function.

\section{Overview of the Algorithm}\label{section:overview-of-the-data-structure}

As in~\cite{Holm2001, Wulff-Nilsen13}, our goal is to restore
Invariant~\ref{invariant:spanning-forest} after each update operation.

In the rest of this section, we provide an overview of the algorithm.
The \underline{underlined parts} of the text refer to primitive
data structure operations supported by Lemma~\ref{lemma-main},
presented in Section~\ref{section:the-backbone-of-the-data-structure}.

\paragraph{The Data Structure.}
The hierarchy \H{} naturally defines a rooted forest (not to be confused with the maximum spanning forest), which is called the \emph{hierarchy forest}, and contains several \emph{hierarchy trees}. We abuse notation and say that $\H$ refers to this hierarchy forest, together with several auxiliary data structures supporting operations on the hierarchy forest.
The nodes in $\H$ are the $i$-components for all $i\in [1,\dmax]$. The roots of the hierarchy trees are nodes in $\hat V_0$, representing 1-components.
The set of nodes at depth $i$ in $\H$ is exactly $\hat V_i$. The set of children of a node $v^i \in \hat V_i$ is
$\{u^{i+1}\in \hat V_{i+1}\ |\ u^i=v^i\}$.
The leaves are nodes in $\hat V_{\dmax} = V$. See Figure~\ref{fig:hierarchy-example} for an example.
The nodes in $\H$ are called \emph{$\H$-nodes}, and the roots are called \emph{$\H$-roots}.

\begin{figure}[ht]
\centering
\includegraphics[width=\textwidth]{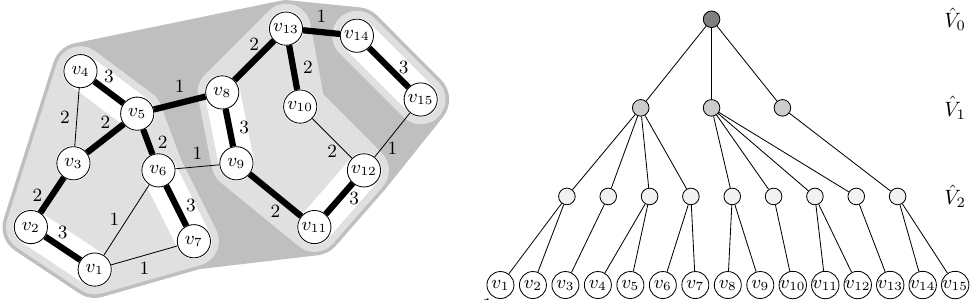}
\caption{An illustration of a graph and the corresponding hierarchy forest $\H$, where $n=15$ and $\dmax = 3$. All thick edges are witness edges and the thin edges are non-witness edges.
Components in $\hat V_1$ have size at most $\floor{\fr{15}{2}} =7$; those in $\hat V_2$ have size at most $\floor{\fr{15}{4}}=3$.}\label{fig:hierarchy-example}
\end{figure}

\subsection{Insertion}\label{section:insertion}

To execute an $\Insert(u, v)$ operation, where $e=\{u, v\}$, the data structure first sets $d_e=1$. If $e$ connects two distinct components in $G$ (which is verified by a connectivity query on $\F$), then the data structure \ul{accesses two $\H$-roots $u^0$ and $v^0$}, \ul{merges $u^0$ and $v^0$} and \ul{$e$ is inserted into $\H$
 (and $\F$) as a $1$-witness edge}.
Otherwise, \ul{$e$ is inserted into $\H$ as a $1$-non-witness edge}.

\subsection{Deletion}\label{section:deletion}
To execute a $\Delete(u, v)$ operation, where $e=\{u, v\}$, the data structure first \ul{removes $e$ from $\H$}.
Let $i=d_e$.
If $e$ is an $i$-non-witness edge, then the deletion process is done.
If $e$ is an $i$-witness edge, the deletion of $e$ could split an $i$-component.
In this case, the deletion algorithm first focuses on finding a replacement edge that has depth $i=d_e$.
In Section~\ref{sec:iteration} we extend our algorithm to find a replacement edge of any depth, while preserving the Maximum Spanning Forest Property of Invariant~\ref{invariant:spanning-forest}.

Prior to the deletion, the edge $\{u^{i}, v^{i}\}$ connected two $(i+1)$-components, $u^{i}$ and $v^{i}$, which, possibly together with some additional $i$-witness edges and $(i+1)$-components, formed a single $i$-component $u^{i-1}=v^{i-1}$ in $\hat G_i$.
If no $i$-non-witness replacement edge exists, then deleting $\{u, v\}$ splits $u^{i-1}$ into two $i$-components.
In order to establish if this is the case, the data structure first \ul{accesses $u^i$, $v^i$ and $u^{i-1}$ in $\H$} and implicitly splits the $i$-component $u^{i-1}$ into two connected components $c_u$ and $c_v$ in $\hat F_i=(\hat V_i, \{\{u^i, v^i\}\ |\ \{u, v\}\in F_i\})$
where $u^i\in c_u$ and $v^i\in c_v$; see Figure~\ref{figure:two-components}(a).
The rest of the deletion process focuses on finding a replacement edge to reconnect $c_u$ and $c_v$ into one $i$-component. This process has two parts, explained in detail below: (1) establishing the two components $c_u$ and $c_v$, and (2) finding a replacement edge. Notice that $c_u$ and $c_v$ do not
necessarily correspond to $\H$-nodes.

\subsubsection{Establishing Two Components}
\label{section:establishing-two-components}

To establish the two components $c_u$ and $c_v$ created by the deletion of $e$, the data structure executes in parallel two graph searches on 
$\hat F_i - \{\{u^i, v^i\}\}$ starting from $u^i$ and $v^i$.
To implement a search, we mark $u^i$ unexplored and insert it into a queue.
We repeatedly remove any unexplored 
$(i+1)$-component $x$ from the queue, mark it explored, and
\ul{enumerate all $i$-witness edges with one endpoint in $x$}.
All new $(i+1)$-components touching these edges are marked unexplored and inserted into the queue.
The two searches are carried out in parallel until one of the 
connected components is fully scanned.
By fully scanning one component, 
the weights of both components are determined, 
since $w(u^{i-1}) = w(c_u) + w(c_v)$.
Without loss of generality, assume that $w(c_u) \leq w(c_v)$, and so by Invariant~\ref{invariant:spanning-forest}, $w(c_u) \leq  w(u^{i-1})/2 \le n/2^{i}$.

\paragraph{Witness Edge Promotions.}
The data structure \ul{promotes all $i$-witness edges touching nodes in $c_u$} and \ul{merges all $(i+1)$-components contained in $c_u$ into one $(i+1)$-component} with weight $w(c_u)$. This is permitted by Invariant~\ref{invariant:weight-property}, since
$w(c_u) \le n/2^{i}$.
The merged $(i+1)$-component has the node $u^{i-1}$ as its parent in $\H$. See Figure~\ref{figure:two-components}$.b$.

\begin{figure}[ht]
\centering
\begin{subfigure}[t]{0.32\linewidth}
\includegraphics[width=\linewidth]{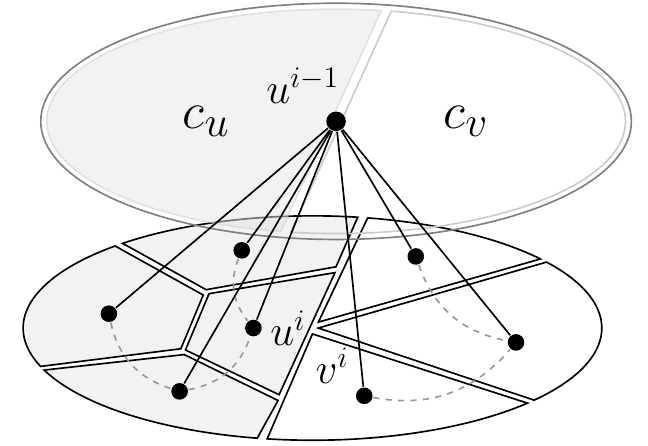}%
\caption{}
\end{subfigure}%
\begin{subfigure}[t]{0.32\linewidth}
\includegraphics[width=\linewidth]{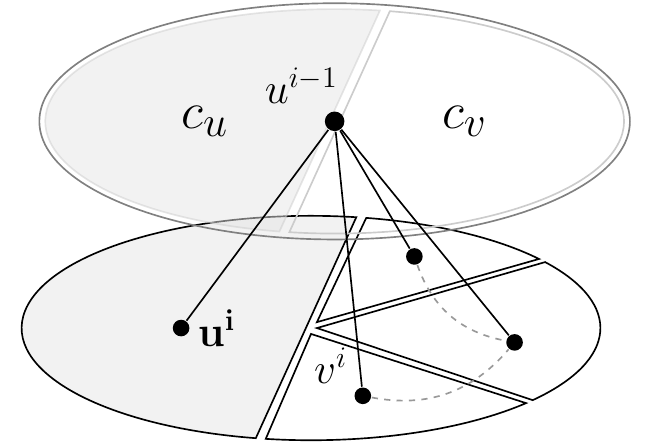}%
\caption{}
\end{subfigure}%
\begin{subfigure}[t]{0.32\linewidth}
\includegraphics[width=\linewidth]{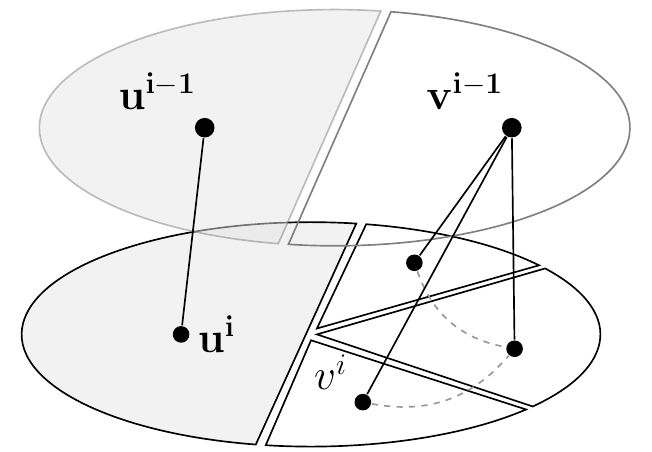}%
\caption{}  
\end{subfigure}
\caption{Illustration of the hierarchy of components at depth $i-1$ and $i$: (a) After identifying two components $c_u$ and $c_v$, it turns out that $c_u$ has smaller weight although it has more $(i+1)$-components. (b) After merging all $(i+1)$-components in the smaller weight component. (c) If no replacement edge is found, then $c_u$ and $c_v$ are two actual connected components in $\hat G_i$ and hence $u^{i-1}$ is split.}\label{figure:two-components}
\end{figure}

To differentiate between versions of components before and after the merges, we use
bold fonts to refer to components after the merges take place.
Thus, the $(i+1)$-component contracted from all $(i+1)$-components inside $c_u$ is denoted $\bf u^i$.
Similarly, the graph $\hat G_i$ after merging some of its nodes is denoted by $\bf{\hat G_i}$.

Having contracted the $(i+1)$-components inside $c_u$ into $\bf u^i$, we now turn our attention to identifying whether the deletion of $e$ disconnects $\bf u^i$ from $c_v$ in $\bf \hat G_i$.

\subsubsection{Finding a Replacement Edge}
\label{section:finding-a-replacement-edge}
Notice that by definition of $\bf \hat G_i$ and $u^{i-1}$, a depth-$i$ edge is a replacement edge in $E$ if and only if it is an $i$-non-witness edge with exactly one endpoint $x\in V$ such that $x^i = \bf u^i$.
To find a replacement edge, the data structure executes one or both of two auxiliary procedures:
the \emph{sampling procedure} and the \emph{enumeration procedure}.

\paragraph{Intuition.} Consider the following two situations.  In Situation A at least a constant fraction of the
$i$-non-witness edges touching $\bf u^i$ have exactly one endpoint touching $\bf u^i$, and are therefore eligible replacement
edges. In Situation B a small $\epsilon$ fraction (maybe zero) of these edges have exactly one endpoint in $\bf u^i$.
If the algorithm magically \emph{knew} which situation the algorithm is in and could sample $i$-non-witness endpoints uniformly at random then
the problem is straightforward to solve:  In Situation A the algorithm iteratively samples an $i$-non-witness endpoint and tests whether the other endpoint
is in $\bf u^i$.  As we will see, each test costs $O(\log n(\log\log n))$ time.
The expected number of samples used to find a replacement edge in this situation is $O(1)$
and so the time cost is charged (in an amortized sense) to the deletion operation.
In Situation B the algorithm enumerates and marks every $i$-non-witness endpoint touching $\bf u^i$.  Any edge with one mark is a
replacement edge and any with two marks may be promoted to depth $i+1$ without violating Invariant~\ref{invariant:weight-property}.
Since a constant fraction of the edges will end up being promoted,
the amortized cost of the enumeration procedure is zero, so long as the enumeration and promotion cost is $O((\log\log n)^2)$ per endpoint.

There are two technical difficulties with implementing this idea.  First, the set of $i$-non-witness edges incident to $\bf u^i$
is a dynamically changing set, and supporting fast (almost-)uniformly random sampling on this set is a very tricky problem.
Second, the algorithm does not know in advance whether the current situation is Situation A or Situation B.  Notice that it is insufficient to draw $O(1)$ random samples
and, if no replacement edge is found, to deduce that the algorithm is in Situation B.
Since the cost of enumeration is so high,
the algorithm cannot afford to \emph{mistakenly} think that it is in Situation B.

\paragraph{Primary and Secondary Endpoints.}
The difficulty with supporting random sampling is dynamic updates:
when $i$-non-witness edges are inserted and deleted from the pool due to promotions, the algorithm responds to each such insertion/deletion with updating $\Omega(\log n)$ parts of the data structure 
that enables fast random sampling.
Thus, the cost of updating each part needs to be 
relatively low in order to obtain the desired time bounds.

Our solution is to maintain two endpoint types for $i$-non-witness edges: \emph{primary} and \emph{secondary}.
A newly promoted $i$-non-witness edge has two \emph{$i$-secondary endpoints} and when an $i$-secondary endpoint $\ang{u, e}$ is enumerated (see the enumeration procedure below), the data structure \emph{upgrades} $\ang{u, e}$ to an \emph{$i$-primary endpoint}.
The motivation for using two types of endpoints is that the algorithm never samples from the set of $i$-secondary endpoints, which are only subject to individual insertions, 
but only the set of $i$-primary endpoints, which are subject to \emph{bulk} inserts/deletes.
The bulk updates to $i$-primary endpoints are sufficiently large (in an amortized sense) to pay for the changes made to the part 
of the data structure that supports random sampling.

Notice that each edge undergoes up to $\dmax$ promotions and up to $2\dmax$ endpoint upgrades.
Since our goal is to obtain an $O(\log n(\log\log n)^2)$ amortized insertion cost, we are able to charge each promotion or 
upgrade $O((\log\log n)^2)$ units of time.

\paragraph{The Sampling Procedure.}
This is the only procedure in our algorithm that uses randomness.
The sampling procedure can be viewed as a two-stage version of Henzinger and Thorup~\cite{Henzinger97}, with some complications due to primary and secondary types.
We give a simple sampling procedure that either provides a replacement 
edge or states that, with high enough probability,
the fraction of $i$-primary endpoints touching $u^i$ that belong to replacement edges is small.

Let $p$ be the number of $i$-primary endpoints touching $\bf u^i$.
The data structure first \ul{estimates $p$ up to a constant factor} and then invokes the \ul{batch sampling test}, which $(1+o(1))$-uniformly samples $O(\log\log p)$ $i$-primary endpoints touching $\bf u^i$.
If an endpoint of a replacement edge is sampled, then the sampling procedure is terminated, returning one of the replacement edges.
Otherwise, the data structure invokes the second \ul{batch sampling test}, which $(1+o(1))$-uniformly samples
$O(\log p)$ $i$-primary endpoints touching $\bf u^i$. The purpose of this step is {\em not} to find a replacement edge,
but to \emph{increase our confidence} that there are actually few replacement edges.
(Since otherwise it is hard to obtain good amortized cost.)
If more than half of these endpoints belong to replacement edges, then the sampling procedure is terminated and one of the replacement edges is returned. Otherwise, the algorithm concludes that the fraction of the non-replacement edges touching $\bf u^i$ is at least a constant,
and invokes the \emph{enumeration procedure}.



\paragraph{The Enumeration Procedure.}
The data structure first \ul{upgrades all $i$-secondary endpoints touching $\bf u^i$ to $i$-primary endpoints},
then 
\ul{enumerates all $i$-primary endpoints touching $\bf u^i$}
and establishes for each such edge the number of its endpoints touching
$\bf u^i$ (either one or both).
An edge is a replacement edge if and only if exactly one of its endpoints is enumerated.
Each non-replacement edge encountered by the enumeration procedure has both endpoints in an $(i+1)$-component, namely $\bf u^i$,
and can therefore be \ul{promoted to a depth $(i+1)$-non-witness edge} (making both endpoints secondary), without violating Invariant~\ref{invariant:spanning-forest}.
As part of the promote and upgrade operations, the algorithm completely rebuilds the part of the data structure supporting random sampling on the $i$-primary endpoints touching $\bf u^i$.

Since the enumeration procedure is only invoked when the algorithm concludes that (before the enumeration process) the fraction of the non-replacement edges touching $\bf u^i$ is at least a constant, the cost of rebuilding the data structure component supporting random sampling is charged to promoting the (sufficiently large number of) non-replacement edges.

\subsubsection{Iteration and Conclusion}\label{sec:iteration}
By the Maximum Spanning Forest Property of Invariant~\ref{invariant:spanning-forest}, the deletion of an edge $e$ can only be replaced by edges of depth $d_e$ or less.
The algorithm always first looks for a replacement edge at the same depth as the deleted edge.
If the algorithm does not find a replacement edge at depth $d_e$ then the algorithm conceptually \emph{demotes} $e$ by setting $d_e\leftarrow d_e - 1$ in order to preserve the Maximum Spanning Forest Property of Invariant~\ref{invariant:spanning-forest},
and continues looking for a replacement edge at the new depth $d_e$.
The demotion is merely conceptual; the deletion algorithm does not actually update $d_e$ in the course of deleting $e$.

\paragraph{Implementation.}
If a depth-$i$ replacement edge $e'$ exists, then $u^{i-1}$ is still an $i$-component and the algorithm \ul{converts $e'$ from an $i$-non-witness edge to an $i$-witness edge}.
Otherwise, $c_u$ and $c_v$ form two distinct $i$-components in $\bf \hat G_i$. In this case,
the data structure \ul{splits $u^{i-1}$ into two sibling nodes} (or two $\H$-roots, if $i=1$): a new node $\bf u^{i-1}$ representing $c_u$ whose only child is $\bf u^i$, and $\bf v^{i-1}$ representing
$c_v$ whose children are the rest of the $(i+1)$-components in $c_v$.
If $i=1$ then the algorithm is done.  
Otherwise, the algorithm sets $i\leftarrow i-1$, conceptually demoting $e$, and
repeats the procedure as if $e$ were deleted at depth $i-1$.

\begin{figure}[ht]
\centering
\includegraphics[width=\linewidth]{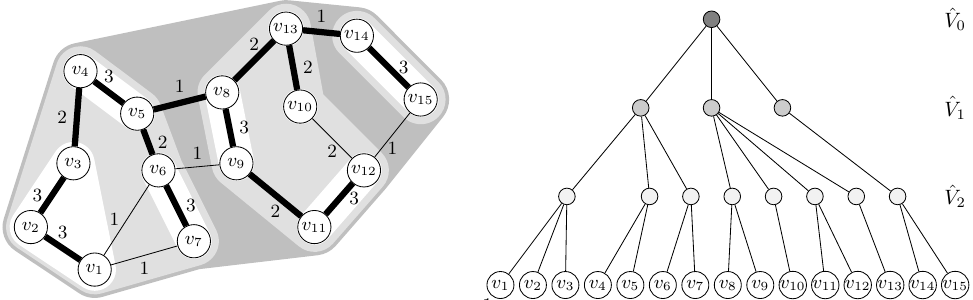}
\caption{After deletion of $\{v_3, v_5\}$ (See Figure \ref{fig:hierarchy-example}.) By identifying $\{v_1, v_2, v_3\}$ to be the smaller weight component, the witness edge $\{v_2, v_3\}$ is promoted to depth 3 and the corresponding nodes in $\hat V_2$ are merged. The edge $\{v_3, v_4\}$ is the replacement edge.}
\end{figure}

\begin{cfigure}[ht]
\centering
\includegraphics[width=\linewidth]{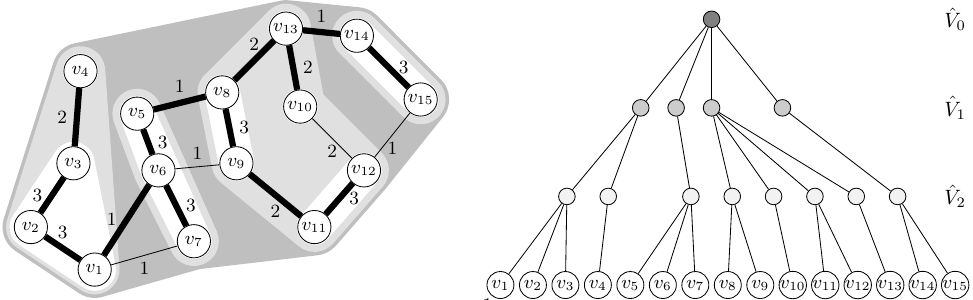}
\caption{After deletion of $\{v_4, v_5\}$ we do the following steps.  (1) Split the node in $\hat V_2$ associated with $v_4$ and $v_5$. (2) Identify that $\{v_5, v_6, v_7\}$ is the smaller weight component. (3) Promote the edge $\{v_5,v_6\}$ to depth 3, merging nodes $v_5^2$ and $v_6^2=v_7^2$.
(4) Fail to find a replacement edge at depth 2, and split the node $v_5^1$.
(5) Find a replacement edge $\{v_1, v_6\}$ at depth 1 and designate it a witness edge.}
\end{cfigure}

\subsection{The Backbone of the Data Structure}\label{section:the-backbone-of-the-data-structure}

Lemma~\ref{lemma-main} summarizes the primitive operations required to execute \Insert\ or \Delete.
Recall that the possible depths are integers in $[1, \dmax]$, and the possible endpoint types are $\const{witness}$ (for endpoints of an $i$-witness edge for some $i$), $\const{primary}$, and $\const{secondary}$.

\begin{lemma}\label{lemma-main}
There exists a data structure
that supports the following operations on \H{} with the following
amortized time complexities (given in parentheses).

\begin{enumerate}[(1)]\compactify
\item{}\label{opr:add-1-non-witness}\label{opr:add-1-witness}\label{opr:remove-i-non-witness}\label{opr:remove-i-witness}%
Add or remove an edge with a given edge depth and endpoint type
$\left(O(\log n (\log\log n)^2)\right)$.

\item{}\label{opr:promote-i-non-witness}\label{opr:promote-i-witness}\label{opr:merge-two-siblings}\label{opr:merge-two-roots}
Given a set $S$ of sibling $\H$-nodes or $\H$-roots, merge them
into a single node ${\bf{u^i}}$, and then promote all $i$-witness edges touching $\bf{u^i}$ to
$(i+1)$-witness edges
$\left(-\Omega((|S|-1)(\log\log n)^2)\right)$.

\item{}\label{opr:convert-i-secondary}Given an $\H$-node $v^i\in \hat V_i$, upgrade all $i$-secondary endpoints touching
$v^i$ to $i$-primary endpoints
$(-\Omega((s-p) (\log\log n)^2)$,
where $p$ and $s$ denote the number of $i$-primary endpoints and $i$-secondary endpoints
touching $v^i$ prior to the upgrade$)$.

\item{}\label{opr:promote-i-primary}%
Given an $\H$-node $v^i\in \hat V_i$ and a subset $S$ of $i$-primary endpoints touching $v^i$, 
promote the endpoints in $S$ to $(i+1)$-secondary endpoints
$(-\Omega((12|S| - p)(\log\log n)^2)$,
where $p$ is the total number of $i$-primary endpoints touching $v^i)$.

\item{}\label{opr:convert-i-non-witness} Convert a given $i$-non-witness edge into an $i$-witness edge
$\left(O(\log n (\log\log n)^2)\right)$.

\item{}\label{opr:split-node} Given two $\H$-nodes $u^{i-1}$ and $u^i$ where $u^i$ is an $\H$-child of $u^{i-1}$,
split $u^{i-1}$ into two sibling $\H$-nodes: one takes $u^i$ as a single $\H$-child and the other
takes the rest of $u^{i-1}$'s former $\H$-children as its $\H$-children
$\left(O((\log\log n)^2)\right)$.

\item{}\label{opr:all-i-witness}\label{opr:all-i-non-witness} Given an $\H$-node $v^i\in \hat V_i$ and a given endpoint type, 
enumerate all endpoints $\Endpoint{u}{e}$ of this type such that 
$d_e=i$ and $u^i=v^i$
$(${}$O(k \log\log n+1)$, where $k$ is the number of enumerated endpoints$)$.

\item{}\label{opr:find-i-component} Given $v^i$, return its $\H$-parent $v^{i-1}$
$\left(O\left(\log\log n + \log\paren{\f{w(v^{i-1})}{w(v^i)}}\right)\right)$.

\item{}\label{opr:counter}Given an $\H$-node $v^i\in \hat V_i$, return a $(1+o(1))$-approximation to the number of $i$-primary endpoints touching $v^i$ $\left(O(1)\right)$.

\item{} (Batch Sampling Test)
\label{opr:sample}\label{opr:check-replacement-edge}
Given an $\H$-node $v^i\in\hat V_i$ and an integer $k$, independently sample $k$ $i$-primary endpoints touching $v^i$
$(1+o(1))$-uniformly at random, and establish for each sampled endpoint whether the other endpoint also touches $v^i$
$($see Remark~\ref{rmrk:batch-sample}$)$.
\end{enumerate}
\end{lemma}

\begin{remark}\label{rmrk:batch-sample}
It should be noted that the time bounds of 
Lemma~\ref{lemma-main} only apply if the operations 
are used to correctly maintain Invariant~\ref{invariant:spanning-forest}.
For example, if we use Operation (5) to create a 
new $i$-witness edge but the set $\mathcal{F}$ (the set of witness edges) 
now contains a cycle, then all bets are off.
Moreover, the worst case cost of the Batch Sampling Test 
operation 
is {}$O(\min(k\log n\log\log n, k + (p+s)\log \log n))$ time,
where $p$ and $s$ are the number of $i$-primary and $i$-secondary endpoints touching $v^i$, respectively.
We analyze the amortized cost of this operation 
only when it is used to find replacement edges and 
maintain Invariant~\ref{invariant:spanning-forest}; 
see Section~\ref{section:the-batch-sampling-test}.
\end{remark}

\begin{remark}\label{rmrk:negative-amortized-cost}
The amortized costs of these operations are with respect
to a potential function (Section~\ref{section:amortized-analysis-of-shortcut-maintenance}).
The most important part of the potential function is that
every upgrade and promotion releases $O((\log\log n)^2)$ units of potential.  Observe that Operations (2,3,4) can have \emph{negative} amortized cost.  Negative amortized costs are not contradictory, 
and they are in fact helpful for paying for the (positive) costs
of operations that occur in conjunction with Operations (2,3,4); 
see Section~\ref{section:proof-of-main-result}.
\end{remark}

The proof of Theorem~\ref{main-result} uses Lemma~\ref{lemma-main} to restore Invariant~\ref{invariant:weight-property}.
The proof itself is mostly a technical recapitulation of the algorithm described
in Section~\ref{section:overview-of-the-data-structure};
for the sake of completeness we provide a full proof in
Section~\ref{section:proof-of-main-result}.

\section{The Main Modules of the Data Structure}\label{section:main-modules}

To support Lemma~\ref{lemma-main}, the data structure utilizes five main modules, some of which depend on each other:
(1) the $\H$-leaf data structure,
(2) local trees,
(3) the notion of an induced $(i,t)$-forest,
(4) shortcut infrastructure, and
(5) approximate counters.
The $\H$-leaf data structure is fairly straightforward and is described in detail in Section~\ref{section:overview:leaf-data-structure}.
We define induced $(i,t)$-forests in Section~\ref{section:the-induced-it-forest}.
A brief overview of the other modules is described in Sections~\ref{section:overview:local-trees},~\ref{section:overview:shortcut-infrastructure}, and~\ref{section:overview:approximate-counters}.
Sections~\ref{section:shortcut-infrastructure-on-H}--\ref{section:amortized-analysis-of-shortcut-maintenance}
provide a detailed explanation of each module.
The general operations involving multiple modules, as well as the proof of Lemma~\ref{lemma-main} are described and analyzed in detail in Section~\ref{section:main-oerations}.

\subsection{The \texorpdfstring{$\H$}{H}-Leaf Data Structure}\label{section:overview:leaf-data-structure}

The $\H$-leaf data structure supports several operations that act on an individual vertex.
Let $v$ be a vertex (an $\H$-leaf), $i\in [1,\dmax]$ be a depth, and $t\in \{\const{witness}, \const{primary},\const{secondary}\}$ be an endpoint type.
The $\H$-leaf data structure supports insertion or deletion of an endpoint (of an edge incident to $v$) with depth $i$ and type $t$.
Moreover, the $\H$-leaf data structure 
supports enumeration of all endpoints 
incident to $v$ with depth $i$ and type $t$,
and selecting one such endpoint uniformly at random.

Supporting these operations in $O(1)$ amortized time (plus time linear in the output) is straightforward.  Simply
pack the endpoints with depth and type $(i,t)$ in a dynamic array.  Dynamic arrays can be implemented deterministically to support incrementing/decrementing the length of the array in $O(1)$ amortized time.

\subsection{The Local Trees}\label{section:overview:local-trees}

A \emph{local tree} is a specially constructed binary tree,
whose root is associated with an $\H$-node $v$ and whose
leaves are associated with the $\H$-children of $v$.
By composing the local trees with $\H$, we can view the result as a single
binary tree of height at most $O(\log n \log \log n)$.
The purpose of this binarization is to provide an efficient 
infrastructure for supporting navigation within $\H$.
The local tree operations are detailed in 
Section~\ref{section:local-trees} and summarized in 
Lemma~\ref{lemma:local-tree-operations}.

\subsection{The Induced \texorpdfstring{$(i,t)$}{(i,t)}-Forest}
\label{section:the-induced-it-forest}

The purpose of the $(i,t)$-forests is to support efficient enumeration of all endpoints of a given type that touch a given $\H$-node.
For a given edge depth $i\in [1, \dmax]$ and endpoint type $t\in \{\const{witness}, \const{primary}, \const{secondary}\}$, an $\H$-leaf $v$ is an \emph{$(i, t)$-leaf} if $v$ has an incident endpoint with depth $i$ and type $t$.
An $\H$-node $v^i\in \hat V_i$ having an $(i, t)$-leaf in its subtree is an \emph{$(i, t)$-root}.
For each $(i, t)$ pair, consider the induced forest $\mathfrak{F}$ on $\H$ by taking the union of the paths from each $(i, t)$-leaf to the corresponding $(i, t)$-root.
An $\H$-node $v$ in $\mathfrak{F}$ is an \emph{$(i, t)$-node} if either
\begin{itemize}\compactify
\item  $v$ is an $(i, t)$-leaf,
\item $v$ is an $(i, t)$-root,
\item $v$ has more than one child in $\mathfrak{F}$ (so $v$ is called an \emph{$(i, t)$-branching node}), or
\item $v$ is an $\H$-child of an $(i, t)$-branching node but $v$ has only one $\H$-child in $\mathfrak{F}$. In this case we call $v$ a \emph{single-child $(i, t)$-node}.
\end{itemize}
Notice that an $(i, t)$-root may or may not be an $(i, t)$-branching node.

For each $(i, t)$-node other than an $(i, t)$-root, define its \emph{$(i, t)$-parent} to be the nearest ancestor in $\mathfrak{F}$ that is also an $(i, t)$-node.
An \emph{$(i, t)$-child} is defined accordingly.
The $(i, t)$-parent/child relation implicitly defines an \emph{$(i,t)$-forest}, which consists of \emph{$(i,t)$-trees} rooted at $\hat V_i$ nodes.
An $\H$-node $v$ has \emph{$(i, t)$-status} if $v$ is an $(i, t)$-node.

\paragraph{Storing $(i,t)$-status.} Each node in $v\in \H$ stores two bitmaps
of size $3\dmax = O(\log n)$ each.
The first indicates for each $(i,t)$ pair
whether $v$ is an $(i,t)$-node, and if so, the second
indicates whether $v$ is an $(i,t)$-branching node or not.

\paragraph{Operations on $(i,t)$-forests.}
A conceptual edge between an
$(i, t)$-node and its $(i,t)$-parent or $(i,t)$-child need not be maintained explicitly.
The two components of our data structure that simulate these edges are the
\emph{shortcut infrastructure} and the \emph{local trees}.
In particular, the shortcut infrastructure supports efficient traversal from a single-child $(i,t)$-node to its unique $(i,t)$-child,
while the local trees support efficient enumeration of all the
$(i,t)$-children of an $(i,t)$-branching node.
Since the implementation of traversal and navigation operations on $(i, t)$-forests utilizes local trees which are introduced and defined in Section~\ref{section:local-trees}, we defer the discussion of
$(i, t)$-forests and their detailed implementation to
Section~\ref{section:proof-lemma-induced-forests} (see Lemma~\ref{lemma-induced-forests}).

\subsection{The Shortcut Infrastructure}\label{section:overview:shortcut-infrastructure}

The purpose of shortcuts is to facilitate a faster traversal from a
single-child $(i,t)$-node to its only $\H$-child.
This traversal costs amortized $O(\log \log n)$ time.
The details and construction of shortcuts are described in Section~\ref{section:shortcut-infrastructure-on-H}.

\subsection{Approximate Counters}\label{section:overview:approximate-counters}

Implementing the sampling operation in Lemma~\ref{lemma-main} reduces to being able to traverse from an $(i, \const{primary})$-branching node to one of its $(i, \const{primary})$-children $v$, where the choice of an $(i,\const{primary})$-child is random with probability that is approximately proportional to the number of $i$-primary endpoints touching $v$.
The implementation of the random choice is supported by maintaining an approximate $i$-counter at each $(i, \const{primary})$-node.
Notice that an $\H$-node could be an $(i, \const{primary})$-node for
several $i$, so there could be several approximate $i$-counters maintained in an $\H$-node.
The advantages of using approximate $i$-counters, as opposed to precise counters, are two-fold. 
First, each approximate $i$-counter uses only $O(\log \log n)$ bits, and so $O(\log n / \log \log n)$ approximate $i$-counters can be packed into a single machine word and be collectively manipulated in $O(1)$ time.  
Second, approximate counters can only take on $(\log n)^{O(1)}$ values, and hence a decrement-only counter can only generate $(\log n)^{O(1)}$ total work throughout its lifetime.
The maintenance of approximate $i$-counters and the sampling algorithm 
are explained in Section~\ref{section:approximate-counters}
and Section~\ref{sect:precision-sampling}.

\ignore{
The precision of an approximate $i$-counter is relative to the depth and weight of the node in which the approximate $i$-counter is stored.
The value at an approximate $i$-counter at an $(i,t)$-node $u$ is computed via a special addition-like binary operation executed on the approximate $i$-counters of the $(i,t)$-children of $u$.
Using the special addition operation introduces a loss in precision; the total loss in precision depends on the height of the arithmetic \emph{formula tree} implicitly formed in the $(i, \const{primary})$-trees,
which are binarized via the local trees.
We formalize this notion of height in Section~\ref{section:approximate-counters} and show
the maximum height is $O(\log n\log\log n)$.
The value stored at each approximate $i$-counter is guaranteed to be a $(1+o(1))$-approximation of the actual number of $i$-primary endpoints touching $v$, which is accurate enough
for the sampling operation, as shown in Section~\ref{section:cost-analysis-sampling-procedure}.
We emphasize that approximate $i$-counters are only stored for $i$-primary endpoints, \emph{not} $i$-secondary endpoints.

The notion of height that is used to prove the approximation guarantee depends on the details of local trees, and so the description of how approximate $i$-counters are updated through changes to $\H$ is deferred to Section~\ref{section:proof-of-lemma-counters} (see Lemma~\ref{lemma-counters}).
}

\section{Shortcut Infrastructure}\label{section:shortcut-infrastructure-on-H}

As described in Section~\ref{section:overview:shortcut-infrastructure}, 
the purpose of shortcuts is to allow for efficient navigation between 
a single-child $(i, t)$-node $u$ and its only $(i,t)$-child $v$. If the 
graph is static, a direct pointer between $u$ and $v$ could be stored 
in the data structure so that $v$ can be directly accessed from $u$. 
The challenge is to maintain 
useful shortcuts in the midst of 
structural updates to $\H$.

\paragraph{$\H$-shortcuts.}
An $\H$-shortcut $\scut uv$ is a data structure connecting an ancestor $u$ to a descendant $v$ in $\H$.
$\H$-shortcuts are only stored between a subset of eligible pairs of ancestor-descendant pairs.  The eligible pairs are determined as follows.
For a positive integer $\ell$, define its \emph{least significant bit index}, denoted by $\LSBIndex(\ell)$, to be the minimum integer $b$ such that $2^b$ divides $\ell$ but $2^{b+1}$ does not.
For an $\H$-node $u$, let $\mathrm{depth}_{\H}(u)$ be the distance from $u$ to the root of the tree in $\mathcal{H}$ that contains $u$.
The \emph{power} of a pair of nodes $u$ and $v$ is defined as
\[
\P(u, v) = \min(\LSBIndex(\mathrm{depth}_{\H}(u) + 1), \LSBIndex(\mathrm{depth}_{\H}(v) + 1)).\footnote{The ``+1'' is included because $\LSBIndex$ is not well defined at zero.}
\]
In order for an \emph{$\H$-shortcut} to exist between $u$ and $v$, any intermediate node $x$ on the path from $u$ to $v$
must have $\LSBIndex(\mathrm{depth}_{\H}(x) + 1) < \P(u,v)$.
If $v$ is the $\H$-child of $u$, then $\P(u,v)=0$ and $\scut uv$
is an eligible $\H$-shortcut, which is called a 
\emph{fundamental $\H$-shortcut}.

The following lemma states that the set of $\H$-shortcuts on an ancestor-descendant path do not cross each other.

\begin{lemma}\label{lemma:covering-prop} For any four distinct $\H$-nodes $x_1, x_2, x_3, x_4$ along a root-to-leaf path in $\H$, it is impossible to have two $\H$-shortcuts $\scut {x_1}{x_3}$ and $\scut {x_2}{x_4}$.
\end{lemma}

\begin{proof}
For $j\in \{1,2,3,4\}$ let $h_j = \mathrm{depth}_{\H}(x_j)+1$, so $h_1 < h_2 < h_3 < h_4$.
Assume the claim is false, and so there exist two $\H$-shortcuts $\scut {x_1}{x_3}$ and $\scut {x_2}{x_4}$.
By definition this implies  $\LSBIndex(h_2) < \LSBIndex(h_3)$ and $\LSBIndex(h_3) < \LSBIndex(h_2)$, a contradiction.
\end{proof}

\paragraph{The covering relationships of $\H$-shortcuts and the poset.}
We say that $\scut ab$ \emph{covers}
$\scut cd$ if $c$ and $d$ are on the path $P_{ab}$ from $a$ to $b$ in $\H$. 
Notice that a shortcut covers itself.
Define $\succeq$ to be the covering partial order:
\[
(\scut ab) \succeq (\scut cd) \mbox{ iff } \scut ab \mbox{ covers } \scut cd.
\]
For any $uv$-path $P_{uv}$ on $\H$, the \emph{maximal covering set} of $P_{uv}$, denoted by
$\ShortCuts^{\H}(u, v)$, is the set of maximal $\H$-shortcuts (with respect to $\succeq$)
among all $\H$-shortcuts having both endpoints on $P_{uv}$.
Figure~\ref{fig:ShortCuts} illustrates $\ShortCuts^{\H}(v_5,v_{14})$ in bold.

\begin{figure}[ht]
\centering
\includegraphics[width=\linewidth]{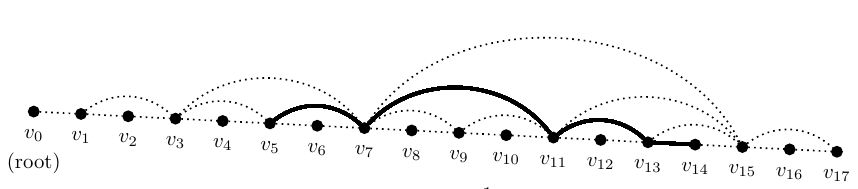}
\caption{\label{fig:ShortCuts}The figure above shows $\ShortCuts^{\H}(v_5, v_{14})$ as an example, where $v_i$ has $\mathrm{depth}_{\H}(v_i)=i$.
The dotted edges are the set of all possible shortcuts.}
\end{figure}

The following lemma bounds the size of $\ShortCuts^{\H}(u, v)$.

\begin{lemma}\label{lemma:largest-shortcuts} For any two nodes $u, v\in \H$ with $u$ an ancestor of $v$, all $\H$-shortcuts in $\ShortCuts^{\H}(u, v)$ form a path connecting $u$ and $v$, and
$|\ShortCuts^{\H}(u, v)|= O(\log \log n)$.
\end{lemma}

\begin{proof}
All $\H$-shortcuts on $P_{uv}$ form a poset, and all fundamental $\H$-shortcuts on $P_{uv}$ form the path between $u$ and $v$.
By Lemma~\ref{lemma:covering-prop}, $\ShortCuts^{\H}(u, v)$ forms a path connecting $u$ and $v$.

The $\H$-shortcuts in $\ShortCuts^{\H}(u, v)$ can be partitioned into two sequences: one with strictly increasing powers and one with strictly decreasing powers.
To see this, notice that for any sequence of consecutive integers, there is a unique largest $\LSBIndex$ value among the sequence.
For any $\H$-node let $q(x) = \LSBIndex(\mathrm{depth}_{\H}(x)+1)$.
Let $v^*$ be the unique $\H$-node on $P_{uv}$ such that $q(v^*) > q(x)$ for all
$x\in P_{uv}\backslash\{v^*\}$.
It is straightforward to see that no $\H$-shortcut on $P_{uv}$ crosses $v^*$ and hence
$\ShortCuts^{\H}(u, v) = \ShortCuts^{\H}(u, v^*)\cup \ShortCuts^{\H}(v^*, v)$.

Now we claim the following: let $P_{v'v}$ be an ancestor-descendant path such that
$q(v')>q(x)$ for all $x\in P_{v'v}\backslash \{v'\}$. Then $\ShortCuts^{\H}(v', v)$ consists of $\H$-shortcuts with strictly decreasing powers.
We prove this claim by induction.  In the base cases the claim is trivially true, when $v'=v$ or $v'$ is the $\H$-parent of $v$.
In general, let $v''$ be the unique node on the path $P_{v'v}$ such that $q(v'') > q(x)$ for all $x\in P_{v'v}\backslash\{v',v''\}$.
The shortcut $\scut {v'}{v''}$ must be in $\ShortCuts^{\H}(v',v)$ since the power of $\scut {v'}{v''}$ is strictly greater than the power of any shortcut on
$P_{v''v}$. 
By the induction hypothesis on $P_{v''v}$, the claim holds.
Thus, all $\H$-shortcuts in $\ShortCuts^{\H}(v^*, v)$ have distinct and decreasing powers.
By symmetry, all $\H$-shortcuts in $\ShortCuts^{\H}(u, v^*)$ 
also have distinct and increasing powers.
Since the maximum depth is $\dmax=\lfloor\log n\rfloor$,
the largest possible power of an $\H$-shortcut is $\lceil\log\log n\rceil-1$.
As a consequence, 
we have $|\ShortCuts^{\H}(u, v)| = O(\log \log n)$.
\end{proof}

\paragraph{$(i,t)$-shortcuts.}
Let $u$ be a single-child $(i,t)$-node and let $v$ be the $(i,t)$-child of $u$,
which by definition must be either an $(i,t)$-branching node or an $(i,t)$-leaf.
The purpose of maintaining $\H$-shortcuts is to allow one to quickly move from $u$ to $v$.
Ideally, the data structure will traverse the $O(\log\log n)$ $\H$-shortcuts in $\ShortCuts^{\H}(u, v)$.
However, forcing all of the $\H$-shortcuts in $\ShortCuts^{\H}(u, v)$ to be maintained by the data structure seems to complicate the process of updating $\H$-shortcuts as $\H$ changes.
In particular, when an $i$-witness 
edge $\{u, v\}$ is deleted, $\H$ goes through several structural changes by merging an ancestor $u^i$ (or $v^i$) with a subset of the 
$\H$-siblings of $u^i$.\footnote{\emph{$\H$-siblings} are $\H$-nodes sharing the same $\H$-parent.}
All $\H$-shortcuts that were connected to $u^i$ (or $v^i$) and those $\H$-siblings need to be updated at the same time.
Since we are fine with $O(\log\log n)$ {\em amortized} time for the traversal, the process of updating shortcuts (due to changes in the hierarchy or the corresponding $(i, t)$-forests) becomes simpler
by allowing a weaker invariant governing which shortcuts are actually present.

\begin{invariant}[$(i, t)$-Shortcuts]\label{invariant:it-shortcuts}
Let $u$ be a single-child $(i,t)$-node and let $v$ be the $(i,t)$-child of $u$.
The $(i,t)$-shortcuts on $P_{uv}$ that are stored by the data structure form a
path connecting $u$ and $v$.
\end{invariant}

When structural changes take place in $\H$, all of the shortcuts that touch the nodes participating in these changes are removed. The cost for removing those shortcuts is amortized over the cost of creating them.
However, once the structural changes are complete, we do not immediately return all the shortcuts back. Instead, the data structure partially recovers enough\footnote{Notice that Invariant~\ref{invariant:it-shortcuts} implies that the shortcut data structure is not required to store all of $\ShortCuts^{\H}(u, v)$ in order for the invariant to hold.} shortcuts to maintain Invariant~\ref{invariant:it-shortcuts},
and then employs a lazy approach in which shortcuts are only added (via a \emph{covering} process) when they are needed.

\begin{figure}[ht]
\centering
\includegraphics[width=5\ThCScm]{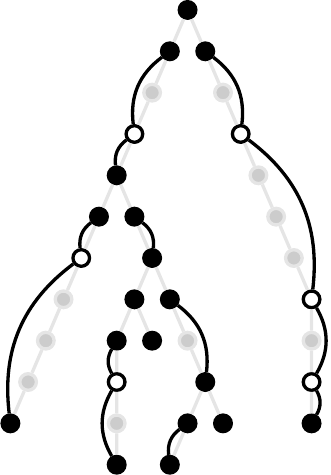}
\caption{An example of an $(i, t)$-tree and its corresponding $(i, t)$-shortcuts: filled circles are $(i, t)$-nodes, and the curved line segments are $(i, t)$-shortcuts.}
\end{figure}

\paragraph{Covering and uncovering.}
Assume Invariant~\ref{invariant:it-shortcuts} holds.
Suppose that the algorithm traverses downward from a single-child
$(i, t)$-node $u$ to its $(i, t)$-child $v$.
If the set of shortcuts used is precisely $\ShortCuts^{\H}(u,v)$ then this traversal costs $O(\log\log n)$ time.
If not, then the algorithm repeatedly \emph{covers} consecutive $(i,t)$-shortcuts (see Section~\ref{section:the-H-shortcut-data-structure} for implementation details)
until the set of $(i, t)$-shortcuts between $u$ and $v$ is exactly $\ShortCuts^{\H}(u,v)$.
We use a potential argument to prove that the amortized cost of traversing from
$u$ to $v$ is $O(\log\log n)$ time; see Section~\ref{section:cost-analysis-lazy-covering}.

There are also certain cases where the structure of $\H$ does not change, but some $(i, t)$-forests do change (for example, whenever an $\H$-leaf gains or loses an $(i, t)$-status).
To support structural changes in $\H$ or in $(i,t)$-forests,
the data structure will at times \emph{uncover} an $(i, t)$-shortcut $s$ of power
$p$ by removing $s$ and adding the two consecutive
$(i, t)$-shortcuts of power $p-1$ that were covered by $s$.
In order to accommodate an efficient uncovering operation, during a covering operation the data structure
continues to store the covered $\H$-shortcuts so that they are readily available when a subsequent \emph{uncover} operation occurs.
The $\H$-shortcuts stored by the data structure that are strictly covered by some $(i, t)$-shortcuts are called \emph{supporting $\H$-shortcuts};
these supporting shortcuts \emph{do not} have $(i,t)$-status.
The $\H$-shortcut $\scut uv$ is always directly accessible from $v$ (the deeper node),
but not necessarily from $u$ (the shallower node).
From the perspective of $v$, $\scut uv$ is called an \emph{upward $\H$-shortcut},
while from the perspective of $u$, $\scut uv$ is called a \emph{downward $\H$-shortcut}.

An upper bound on the number of $\H$-shortcuts that need to be stored at each $\H$-node is captured by the following straightforward corollary.  (Recall that the algorithm \underline{does not} store shortcuts between an $(i,t)$-branching node and its $(i,t)$-children.)

\begin{corollary}\label{corollary:it-shortcuts} Assume Invariant~\ref{invariant:it-shortcuts} holds for all pairs of nodes in $\H$.
Then for each node $v\in \H$, and each $(i,t)$ pair, there is at most one downward $(i, t)$-shortcut and at
most one upward $(i, t)$-shortcut at $v$.
\end{corollary}

\paragraph{Sharing shortcuts.}
An $\H$-shortcut $\scut uv$ that is an $(i, t)$-shortcut could also be an $(i', t')$-shortcut when
$(i,t)\neq (i',t')$.
Similarly, a supporting shortcut for some $(i, t)$-shortcut could also be an $(i', t')$-shortcut when
$(i,t)\neq (i',t')$.
The data structure stores at most one copy of any $\H$-shortcut even if there are many $(i,t)$ pairs that use it.
The maximum number of distinct $\H$-shortcuts touching a given ancestor-descendant path is bounded by the following lemma.
(Recall that a \emph{stored shortcut} is either an $(i,t)$-shortcut for
some $(i,t)$, or a supporting shortcut, which may have no $(i,t)$-status.)

\begin{lemma}\label{lemma:shortcuts-touching-a-path}
Consider any node $v$ in $\H$. The total number of stored shortcuts joining an ancestor of $v$ to another node
is $O(\log n\log\log n)$.
In particular, the number of distinct fundamental $(i,t)$-shortcuts having one endpoint 
at an ancestor of $v$ is $O(\log n)$.
Moreover, the number of $\H$-shortcuts having both endpoints at ancestors of $v$ is $O(\log n)$.
\end{lemma}

\begin{proof}
For a given path $P$, an $\H$-shortcut $\scut uv$ is said to be \emph{deviating} if exactly one of its endpoints is on $P$.

Let $P$ be the path from $v\in \H$ to its $\H$-root. For each edge depth $i$ and type $t$, at most one $(i, t)$-shortcut is deviating from $P$, and each such shortcut has at most $O(\log \log n)$ supporting shortcuts with exactly one endpoint on $P$.
(Recall that $(i,t)$-shortcuts form paths from \emph{single-child} $(i,t)$-nodes to their $(i,t)$-child.  Branching $(i,t)$-nodes have no $(i,t)$-shortcuts leading to descendants.)
Thus, for each $(i, t)$ pair, at most one \emph{fundamental} 
$(i, t)$-shortcut deviates from $P$.
All $\H$-shortcuts connecting $\H$-nodes on $P$ form a laminar set, and so there are at most $2\dmax = O(\log n)$ such $\H$-shortcuts.
Thus, the total number of stored shortcuts with one endpoint in $P$ is $O(\log n\log\log n)$, and the total number of distinct
fundamental $(i, t)$-shortcuts with one endpoint on $P$ is $O(\log n)$.
\end{proof}

In the rest of this section, 
we describe how $\H$-shortcuts are 
stored.

\subsection{The \texorpdfstring{$\H$}{H}-shortcut data structure}\label{section:the-H-shortcut-data-structure}

\paragraph{Information stored at $\H$-nodes.}
Due to Corollary~\ref{corollary:it-shortcuts}, 
every node in $\H$ has at most $3\dmax+1 = O(\log n)$
downward $(i,t)$-shortcuts at any given time.
Each node $u$ stores an array $\Dsc_u$ of size at most $3\dmax+1$ storing all downward $(i,t)$-shortcuts, together with a bitmap $\Occ_u$ indicating which array slots of $\Dsc_u$ are in use.\footnote{Notice 
that when the data structure 
allocates the array $\Dsc_u$, 
it is assumed to contain arbitrary values.
One can only tell which values are meaningful 
and how to interpret them via the 
$\Occ_u$ and $\Dptr_u$ arrays.}
The size of $\Dsc_u$ is chosen to be exactly enough for storing pointers to $(i,t)$-shortcuts for all possible $(i,t)$ pairs as well as one additional slot for temporary use during promotions/upgrades.
However, a single shortcut may be shared by many $(i,t)$ pairs.
In order to support fast access from $u$ to its downward $(i,t)$-shortcut, each node stores a \emph{local dictionary} which is an array $\Dptr_u$ storing, for each $(i,t)$ pair, a $(\log \log n+2)$-bit index to the location in $\Dsc_u$ of the appropriate downward $\H$-shortcut, i.e.,
\[
\Dsc_u[\Dptr_u[i,t]] \mbox{ points to an $(i, t)$-shortcut leaving $u$, if such a shortcut exists.}
\]

Notice that for an $\H$-node and a power $p$, there is at most one upward $\H$-shortcut from $v$ with power $p$. Thus, each node $v$ maintains an array $\Usc_v$ of $O(\log\log n)$ pointers to shortcuts,
sorted by power, to the upward supporting $\H$-shortcuts of $v$.
Moreover, at each node $v$ the data structure stores a $(3\dmax+1)$-length array $\Uptr_v$ of $O(\log\log\log n)$-bit integers for each $(i,t)$ pair. Thus, the upward $(i, t)$-shortcut $\scut xv$ is accessed in $O(1)$ time, i.e.,
\[
\Usc_v[\Uptr_v[i,t]] \mbox{ points to an $(i, t)$-shortcut entering $v$, if such a shortcut  exists.}
\]

Notice that each entry in the $\Dptr_u$ and $\Uptr_v$ arrays is represented with $O(\log\log n)$ bits,
and there are $O(\log n)$ $(i, t)$ pairs. These entries are packed into $O(\log\log n)$ memory words
so that the data structure is able to update the entire array efficiently
via lookup tables in $O(\log\log n)$ time.

The following lemma summarizes how shortcuts are used to support various operations needed
locally in one $\H$-node.

\begin{lemma}\label{lemma:shortcuts-operation}
The following operations are supported via shortcut information stored at nodes
(worst case time in parentheses).

\begin{itemize}\compactify
\item Given $\scut uv$ and a bitmap $b$ of length $3\dmax+1$, add $\scut uv$ as an $(i,t)$-shortcut for all $(i,t)$ pairs indicated by $b$ $(${}$O(\min\{|b|+1, \log\log n\})$ where $|b|$ is the number of $1$s in $b${}$)$.
\item Given $\scut uv$ and a bitmap $b$ of length $3\dmax+1$, remove the $(i, t)$-shortcut status from $\scut uv$ for all $(i,t)$ pairs indicated by $b$ $(${}$O(\min\{|b|+1, \log\log n\})${}$)$.
\item Given $u\in \H$ and an $(i,t)$ pair, return the $(i, t)$-downward $\H$-shortcut at $u$ or report that such a shortcut does not exist $\left(O(1)\right)$.
\item Given $v\in \H$ and an $(i,t)$ pair, return the $(i, t)$-upward $\H$-shortcut at $v$ or report that such a shortcut does not exist $\left(O(1)\right)$.
\item Given $u\in \H$, return the index of an empty slot in $\Dsc_u$ $(O(1))$.
\item Given $u\in \H$, enumerate all indices of used locations in $\Dsc_u$ $(O(k+1)$ where $k$ is the number of the enumerated indices$)$.
\end{itemize}
\end{lemma}

\begin{proof}[Proof Sketch]
   The proof of the lemma is straightforward using bitwise operations or $O(n^\epsilon)$-size lookup tables for operations on $\Dsc_u$, $\Dptr_u$, $\Occ_u$, $\Uptr_v$, and $\Usc_v$.  For example, $\Occ_u$ is a $(3\log n+1)$-bit vector. We partition it into $3\epsilon^{-1}$ segments of $\epsilon\log n$ bits, and can search for a zero in each segment in $O(1)$ time with a table lookup.
\end{proof}

\paragraph{Information stored at shortcuts.}
An $\H$-shortcut $\scut uv$ is an $O(1)$-word data structure storing the following information:
\begin{itemize}

\item $\P(u,v)$: the \emph{power} of the shortcut,
\item Pointers to $u$ and $v$,
\item The index $j$ in $\Dsc_u$ where $\scut uv$ is stored, or $\bottom$ if $\scut uv$ is not stored in $\Dsc_u$,
\item A $3\dmax+1$ length bitmap $b_{\scut uv}$ containing one bit for each $(i,t)$ pair (called the $(i,t)$-bit)
indicating whether $\scut uv$ is an $(i,t)$-shortcut, and
\item If $\P(u,v)>0$ then $\scut uv$ stores pointers to the two
supporting shortcuts with power $\P(u,v)-1$ that $\scut uv$ covers.
\end{itemize}

\begin{lemma}\label{lemma:lazy-cover-operation}

The $\H$-shortcut data structure supports the following operations (worst case time in parenthesis):
\begin{enumerate}\compactify
\item  (Uncovering) Given an $(i,t)$ pair and an $(i,t)$-shortcut $\scut{u}{v}$ that is not a fundamental $\H$-shortcut, uncover $\scut{u}{v}$ and convert the two supporting shortcuts of $\scut{u}{v}$ into $(i,t)$-shortcuts $(O(1))$.
\item  (Traversal and Covering) Assume Invariant~\ref{invariant:it-shortcuts} holds for all $\H$-nodes with depth $\ge i$. Given a single-child $(i,t)$-node $u$ whose $(i,t)$-child is $v$,
traverse from $u$ to $v$ via $(i,t)$-shortcuts while guaranteeing that after the traversal is completed, the set of $(i,t)$-shortcuts between $u$ and $v$ is
exactly $\ShortCuts^{\H}(u,v)$, preserving Invariant~\ref{invariant:it-shortcuts} for all $\H$-nodes with depth $\ge i$ \ $(O(k + \log\log n)$, where $k$ is the number of $(i, t)$-shortcuts 
covered during the traversal$)$.
\end{enumerate}
\end{lemma}

\begin{proof}
Part 1.
Suppose the algorithm uncovers a given $(i,t)$-shortcut $\scut uv$  that is not fundamental, meaning $\scut uv$  has power $p>0$.
The algorithm sets $b_{\scut uv}[i,t]=0$, follows the two pointers from
$\scut uv$ to its supporting power-$(p-1)$ $\H$-shortcuts $\scut ux$ and $\scut xv$, and sets
$b_{\scut ux}[i,t]=b_{\scut xv}[i,t]=1$.
The algorithm also updates in a straightforward manner some local information in all affected nodes $\{u,v,x\}$ in $O(1)$ time.
To be specific, the algorithm
(i) checks whether $\scut ux$ and $\scut xv$ are already stored in $\Dsc_u$ and $\Dsc_x$ by inspecting $\scut ux$ and $\scut xv$.
(ii) If not, the algorithm finds empty slots in $\Dsc_u$ and/or $\Dsc_x$ via the bitmaps $\Occ_u$ and $\Occ_x$,
which indicate which slots in $\Dsc_u$ and $\Dsc_x$ are available, and updates $\Dptr_u[i,t]$ and/or $\Dptr_x[i,t]$.
(iii) The algorithm sets $\Dsc_u[\Dptr_u[i, t]] = \scut ux$; $\Dsc_x[\Dptr_x[i, t]] = \scut xv$; $\Uptr_x[i, t]=\P(u, x)$; and $\Uptr_v[i, t]=\P(x, v)$.
(iv) If $b_{\scut uv}$ is all 0, i.e., $\scut uv$ is no longer an $(i',t')$-shortcut for any $(i',t')$ pair,
the algorithm frees the slot storing $\scut uv$ in $\Dsc_u$ by unsetting the corresponding bit in $\Occ_u$, 
then updates $\scut uv$ to reflect that $\scut uv$  is no longer stored in $\Dsc_u$.

\begin{remark}
After step (iv), it may be that $b_{\scut uv}=0$ \emph{and} $\scut uv$ is not a supporting shortcut for any higher-power
$(i',t')$-shortcut.  If this is the case, it is fine to delete $\scut uv$ (and update $\Usc_v$ and $\Uptr_v$ appropriately).
In our implementation the algorithm has no means to check whether $\scut uv$ is a necessary supporting shortcut,
and so the algorithm keeps $\scut uv$ allocated.
Notice that in the worst case there are $\Theta(n\log\log n)$ stored shortcuts, so keeping spurious shortcuts around does
not affect the overall space usage of the data structure.
\end{remark}

Continuing with the proof (Part 2), we can move from $u$ to its $(i,t)$-child by starting at $u$ and following downward $(i,t)$-shortcuts until an $(i,t)$-node is reached. 
During this traversal, if there are two consecutive $(i, t)$-shortcuts $\scut {x}{y'}$ and $\scut {y'}y$ 
with the same power $p$ and
\[
\LSBIndex(\mathrm{depth}_{\H}(y')+1) < \min\Big(\LSBIndex(\mathrm{depth}_{\H}(x)+1),\; \LSBIndex(\mathrm{depth}_{\H}(y)+1)\Big),
\]
then the data structure covers the two shortcuts with the $\H$-shortcut $\scut xy$ having power $p+1$.
This is done as follows.

First the algorithm checks whether $\scut xy$ already exists, by testing if $\Usc_y[\P(x,y)]$ stores a pointer
to $\scut xy$ or not. If $\scut xy$ already exists then $\scut xy$ is accessed through $\Usc_y$,
and if not then $\scut xy$ is created and a pointer to $\scut xy$ is added to $\Usc_y$.

Next, the algorithm sets the $(i, t)$-bit in $b_{\scut xy}$ to $1$
and sets the $(i, t)$-bits in $b_{\scut {x}{y'}}$ and
$b_{\scut {y'}y}$ to $0$.  
If $b_{\scut x{y'}}=0$ (resp. $b_{\scut {y'}y}=0$), the algorithm removes its index from $\Dsc_x$ (resp. $\Dsc_{y'}$) by unsetting the corresponding bits in $\Occ_x$ (resp. $\Occ_{y'}$). The algorithm also updates $\Dsc_x$ and $\Dptr_x$ so that $\Dsc_x[\Dptr_x[i, t]] = \scut xy$ is accessible from $x$.

Covering $\scut xy$ may create the opportunity to cover another shortcut
$\scut {x'} y$ of the next higher power.
The data structure uses $\Usc_x$ to access the upwards shortcut $\scut {x'} x$ with power $p+1$.
If $\scut {x'}x$ exists and is also an $(i,t)$-shortcut then the data structure covers $\scut {x'}x$ and $\scut xy$
with $\scut {x'} y$, and recursively looks to see if there are more shortcuts to cover at power $p+2$, and so on.

It is straightforward to verify that at the end of the traversal the set of $(i,t)$-shortcuts connecting
$u$ and $v$ is exactly $\ShortCuts^{\H}(u,v)$. The time for traversing the path
is $O(k + |\ShortCuts^{\H}(u, v)|)$, which is $O(k+\log\log n)$ where $k$ is the number of
$(i, t)$-shortcuts being covered during the traversal.
\end{proof}

In Section~\ref{section:cost-analysis-lazy-covering} we show that by
defining the potential function to be the number of all $(i, t)$-shortcuts
that \emph{could} be covered but are not yet covered, this operation has amortized cost $O(\log\log n)$ time.

\subsection{Maintaining Invariant~\ref{invariant:it-shortcuts} Through Structural Changes to \texorpdfstring{$\H$}{H}}\label{section:uncovering-a-path}

\begin{figure}[ht]
   \centering
   \includegraphics[width=2.5\ThCScm]{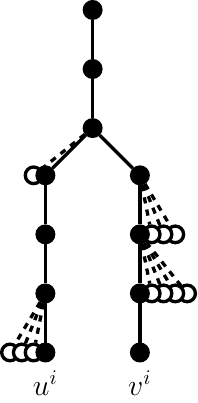}
   \caption{After deleting an $i$-witness edge $\{u, v\}$, all affected $\H$-nodes are on at most two paths.
   In the course of looking for a replacement for $\{u,v\}$, we will merge a collection of siblings into one at
   each depth between $i$ and the depth where the replacement edge is found.
   The dashed lines illustrate the effect of merging siblings.}\label{fig:two-paths}
   \end{figure}

The shortcut infrastructure is very sensitive to the merge operation (e.g., Operation~\ref{opr:merge-two-siblings} in Lemma~\ref{lemma-main}). 
In particular, when an $i$-witness edge $\{u,v\}$ is deleted, $\H$ goes through several structural changes by merging an ancestor of $u^i$ (or $v^i$) with a subset of its $\H$-siblings. These merges require 
updating the shortcut infrastructure, which seems
to be a very complicated task when supporting these types of changes. 
Specifically, we need to employ a special strategy that ensures Invariant~\ref{invariant:it-shortcuts} holds after the entire \Delete{} operation.

In order to provide an efficient implementation, observe that during such a single deletion, all merged $\H$-nodes (and their appropriate $\H$-siblings) end up being on the paths between $u^i$ and $v^i$ and their respective $\H$-roots.  See Figure~\ref{fig:two-paths}. Thus, we are able to employ the following strategy.

First, at the beginning of the delete operation, the algorithm completely uncovers and removes \emph{all} $\H$-shortcuts that touch $\H$-nodes on the two paths. In particular, by Lemma~\ref{lemma:shortcuts-touching-a-path}, the algorithm removes 
(1) $O(\log n)$ fundamental shortcuts, 
(2) $O(\log n)$ shortcuts with both endpoints on the path, and (3) $O(\log\log n)$ deviating shortcuts from each path for each $(i, t)$ pair.
Recall that deviating shortcuts have one endpoint on the $\H$-path in question.

After removing these $\H$-shortcuts, Invariant~\ref{invariant:it-shortcuts} no longer holds for pairs of $\H$-nodes where at least one node is on the affected paths.
However, these are shortcuts with $(i', t')$-status for some $i' < i$, and so during the deletion operation at depth $i$ we never use such shortcuts. Hence, removing them does not affect the other operations that take place during the edge deletion process at depth $i$.


Lemma~\ref{lemma:uncover-a-path} summarizes
the operations that remove and 
restore shortcuts along paths in $\H$,
which are used to guarantee that 
Invariant~\ref{invariant:it-shortcuts} holds after the deletion operation terminates.
Since the implementation requires 
interaction with 
the local trees, we defer its proof to Section~\ref{section:uncover-a-path}.

\begin{lemma}\label{lemma:uncover-a-path}
   The data structure supports the following operations on $\H$ with amortized time cost (in parenthesis).  Given an $\H$-node $v$:
   \begin{itemize}\compactify
   \item Uncover and remove every $\H$-shortcut that is touching any node that is an ancestor of $v$ $(O(\log n(\log\log n)^2))$.
   \item Given $v$, its $\H$-parent $u$, and a bitmap $b$, add a fundamental $\H$-shortcut $\scut{u}{v}$ for all $(i, t)$ pairs indicated by $b$ $(O(\log\log n))$.
   \item Add all fundamental $\H$-shortcuts between consecutive ancestors of $v$
   that are $(i,t)$-shortcuts for at least one $(i, t)$ pair
   $(O(\log n \log\log n))$.
   \item Assume Invariant~\ref{invariant:it-shortcuts} holds.
   For all $(i, t)$ pairs, cover all $(i, t)$-shortcuts having both endpoints at ancestors of $v$ $(O(\log n \log\log n))$. 
   \end{itemize}
   \end{lemma}


\section{Implementation of Approximate Counters}\label{section:approximate-counters}\label{section:the-sampling-procedure}

In this section, we describe how approximate $i$-counters are implemented.
Without loss of generality we assume that the input graph $G$ is simple.
Hence, all approximate $i$-counters are only required to represent a $(1+o(1))$-approximation of integers in
the range $[0, n^2]$.

\subsection{Approximate Counters}
Each $(i, \const{primary})$-leaf $\ell$ maintains the exact number of $(i, \const{primary})$-endpoints touching $\ell$.
The precise number of $(i, \const{primary})$-endpoints in a subtree of any $(i, \const{primary})$-node $v$ could be computed \emph{exactly} using a formula tree defined by the induced  
$(i, \const{primary})$-tree rooted at $v$ where the value at each vertex is the sum of the values of its children.
(Because the local trees are binary, the induced tree is also binary, and has height $O(\log n\log\log n)$.)
If one were to use such a strategy, then every $\H$-node has the potential of storing $O(\log n)$ counters, where each counter uses $O(\log n)$ bits, for a total of $O(\log n)$ words. Thus, splitting and merging vertices may cost $\Theta(\log n)$ time each, which is too expensive for our purposes.

Instead, the data structure efficiently maintains \emph{approximate} $i$-counters for nodes in $\H$ with a multiplicative approximation guarantee of $(1+o(1))$ using only $O(\log \log n)$ bits per approximate $i$-counter.

\paragraph{The structure of an approximate counter.}
Let $\beta = 2$ be a parameter that controls the quality of the approximation.
Each approximate counter $\hat C$ is defined by a pair $(m, e)$ composed of a \emph{mantissa} $m\in \{0, 1\}^{\beta\log\log n}$ and an \emph{exponent} $e\in\{0, 1\}^{\log\log n+1}$. The \emph{floating point representation} of $\hat C$ concatenates $m$ and $e$ into a length $(\beta+1)\log\log n+1$ bit string. The \emph{integer representation} of $\hat C$ is $m2^e$, where we treat the mantissa part and the exponent part as unsigned integers. Notice that an approximate counter represents up to $2(\log n)^{\beta + 1}$ different integers.
From the definition above, an integer $C\in [0, n^2]$ is approximated by $\hat C = (m, e)$ where $m$ is the first $\beta\log\log n$ bits of the binary representation of $C$ and $e$ is the number of truncated bits.

\paragraph{Special addition operation.}
When computing the addition of two values $a$ and $b$ represented by two approximate counters,
the result $a+b$ is rounded \emph{down} to the nearest possible approximate counter value.
Notice that this kind of addition is not associative.
We denote the operation of adding two approximate counters by $a \boxplus b$.
The precision guarantee of $\boxplus$ is summarized in the following lemma.

\begin{lemma}
Let $a$ and $b$ be two approximate counter values represented by approximate counters. Then $a \boxplus b$ satisfies:
\[
(1-\log^{-\beta}n) (a+ b)\le a \boxplus b \le a + b.
\]
\end{lemma}

\begin{proof}
   Let $C=a+b$. Then by definition $\hat{C}=(m, e)$ keeps the first $\beta\log\log n$ bits of the binary representation of $C$. The difference between $C$ and $\hat{C}$ is therefore strictly less than $(\log^{-\beta}n)C$.
   Thus, $\hat{C} \ge (1-\log^{-\beta} n)C$.
\end{proof}

\paragraph{Approximation guarantee and the formula tree.}
Using approximate counters with the $\boxplus$ operation instead of exact counters
creates a loss in precision which depends on the height of the arithmetic formula tree.
Recall that the height of a formula tree is always bounded by $O(\log n \log \log n)$ where the $\log \log n$ factor is due to the local trees.
In order to bound the loss of precision we use a function $H(v)$
which expresses the \emph{maximum possible} height of $v$ in any formula tree.  See Section~\ref{sect:precision-sampling} 
for more on why $H(\cdot)$ is defined this way.
\begin{definition}
Let $v$ be an $\H$-node.
Let $j$ be the depth of $v$ in $\H$. Then
\[
H(v) = (\dmax-j)\cdot O(\log\log n)%
+\lfloor\log (w(v))\rfloor.
\]
\end{definition}
Notice that $H(v) = O(\log n\log\log n)$.
The following invariant relates the precision of approximate counters to the function $H$.
The maintenance of Invariant~\ref{invariant:approximate-counters} is addressed in Lemma~\ref{lemma-counters}, 
which is proved in Section~\ref{section:proof-of-lemma-counters}.
\begin{invariant}[Precision of Approximate Counters]\label{invariant:approximate-counters}
Let $v$ be an $\H$-node and let $C_i(v)$ be the precise number of $i$-primary endpoints touching $v$.
If $v$ is an $(i,\const{primary})$-node then $v$ stores an approximate $i$-counter $\hat C_i(v)$,
where
\[
\left(1-(\log^{-\beta} n)\right)^{H(v)} C_i(v) \le \hat C_i(v) \le C_i(v).
\]
\end{invariant}
Thus, if Invariant~\ref{invariant:approximate-counters} holds with $\beta=2$, then for any $\H$-node $v$,
$$\hat C_i(v) \ge \left(1-(\log^{-2} n)\right)^{H(v)} C_i(v) = \left(1-(\log^{-2} n)\right)^{O(\log n \log \log n)} C_i(v) = (1-o(1)) C_i(v),$$ and so $\hat C_i(v)$ gives the desired approximation.

\paragraph{Packing $O(\log n)$ Approximate Counters.}
Each node in $\H$ stores $\dmax = \log n$ approximate counters. These counters are stored in $O(\log\log n)$ words by packing $O(\log n/\log \log n)$ approximate counters in the floating pointer representation into each word.
With the aid of lookup tables of size $O(n^\epsilon)$, the following lemma is straightforward.

\begin{lemma}\label{lemma:approximate-counters-addition}
The following operations are supported on approximate counters (worst case time in parentheses):

\begin{itemize}\compactify
\item  Given an $\H$-node $v$ and a depth $i$, update/return the approximate $i$-counter stored at $v$ \ $(O(1))$.
\item Given the floating point representation of an approximate counter, return its integer representation \ $(O(1))$.
\item Given the integer representation of an approximate counter, return its floating point representation \ $(O(1))$.
\item Given two approximate counters $a$ and $b$, return $a\boxplus b$ \ $(O(1))$.
\item Given two arrays of $O(\log n)$ approximate counters packed into $O(\log\log n)$ words, return
their coordinate-wise sum, packed into $O(\log\log n)$ words \ $(O(\log \log n))$.
\end{itemize}
\end{lemma}

\begin{proof}[Proof Sketch]
The first four operations use bitwise operations in a straightforward manner.
The fifth operation uses $O(n^\epsilon)$-size lookup tables to support a query in $O(\log\log n)$ time; see Section~\ref{section:lookup-tables}.
\end{proof}

\paragraph{Summary of operations.}
The main lemma summarizing operations related to approximate $i$-counters is given next.
\begin{lemma}\label{lemma-counters}
   There exists a data structure that maintains approximate $i$-counters on $\H$ while maintaining Invariant~\ref{invariant:approximate-counters} and supporting the following operations with the following amortized time complexities (in parentheses):
   \begin{itemize}\compactify
   \item{} Update the approximate counters to reflect a change in the number of $(i,\const{primary})$-endpoints at a given $\H$-leaf
   $(${}$O(\log n(\log \log n)^2)${}$)$.

   \item{} Given an $(i, \const{primary})$-tree $\T$ rooted at $v^i$, rebuild approximate $i$-counters for all $(i, \const{primary})$-nodes
   in $\T$ to restore Invariant~\ref{invariant:approximate-counters} for those nodes $(${}$O(|\T| (\log \log n)^2)${}$)$.

   \item{} When merging two sibling $\H$-nodes, compute the approximate $i$-counters for all $i\in [1,\dmax]$ at the merged node
   $(O(\log \log n))$.

   \item{} When splitting an $\H$-node into two sibling $\H$-nodes, compute the approximate $i$-counters for all $i\in [1,\dmax]$ at the two sibling nodes $(O(\log \log n))$.
   \end{itemize}
   \end{lemma}

The proof of Lemma~\ref{lemma-counters} depends on the implementation of local trees, which we provide in Section~\ref{section:local-trees}.
Thus, the proof of Lemma~\ref{lemma-counters} is deferred to Section~\ref{section:proof-of-lemma-counters}.

\section{Local Trees}\label{section:local-trees}

The purpose of the local tree $\L(v)$ is to connect 
an $\H$-node $v$ with its $\H$-children while supporting various
operations. 
A local tree is composed of a three-layer binary tree 
and a special binary tree called the {\emph{buffer tree}}. 
The three-layer binary tree is composed of a {\emph{top layer}}, a {\emph{middle layer}} and a {\emph{bottom layer}}.
See Figure~\ref{figure:local-tree} for an illustration.

\begin{figure}[ht]
\centering
\includegraphics[width=0.8\linewidth]{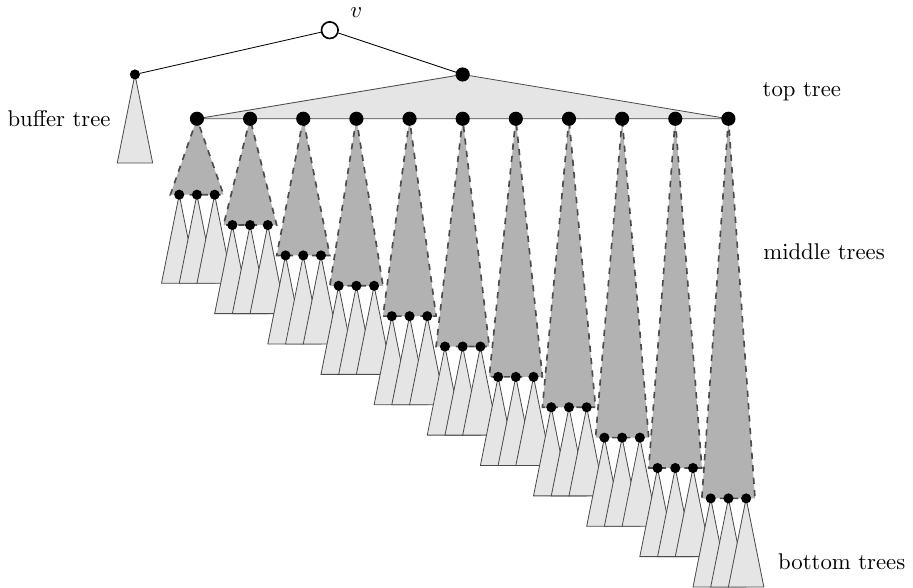}
\caption{An example to a local tree $\L(v)$ associated with $v$.}
\label{figure:local-tree}
\end{figure}

\begin{itemize}
    \item The bottom layer is composed of \emph{bottom trees}, each having at most $2\log^\alpha n$ leaves and 
    height $O(\log\log n)$, $\alpha=O(1)$ to be calculated later.

\item The middle layer is composed of \emph{middle trees} 
such that all bottom tree roots are middle tree leaves.
The \emph{weight} of a  node $x$ in $\L(v)$, 
denoted by $w(x)$, is defined to be the sum of all weights of
$\H$-children of $v$ in the subtree of $x$, and the \emph{rank} of $x$ is defined to be $\mathrm{rank}(x) = \floor{\log w(x)}$.
The weights 
are explicitly maintained only for 
nodes in either bottom or buffer trees.
The middle trees are weight balanced binary trees with respect to $w(\cdot)$.
The algorithm maintains the invariant that there are never 
more than $O(\log n)$ middle tree roots in a local tree.  

\item The top tree\footnote{Not to be confused with the \emph{top tree} dynamic tree data structure of Alstrup, Holm, Lichtenberg, and Thorup~\cite{AlstrupHLT05}.} 
is a mergeable, $O(\log\log n)$-height tree 
whose leaves are middle tree roots.  
Its purpose
is merely to gather up all middle trees 
within a single tree, while increasing the overall height of the local tree by only $O(\log\log n)$.
\end{itemize}

\paragraph{Local tree roots and local tree leaves.}
The root of $\L(v)$ has two children: 
the root of the buffer tree and the root of the top tree.
The root of $\L(v)$ also has a pointer pointing to $v$ in $\H$. 
When a new $\H$-node $v\in \H$ is created, $\L(v)$ is initially empty.

\paragraph{$\H$-node representatives.}
Each $\H$-child $x$ of $v$ is not in $\L(v)$ as such, 
but is present through a \emph{representative} $\ell_x$, 
which is a leaf in $\L(v)$.
We distinguish $x$ from $\ell_x$ because they have 
different characteristics and store different information.

The local tree leaf $\ell_x$ stores a pointer to $x\in\H$, 
the weight of $x$, a parent pointer, approximate counters,
and a bitmap maintaining 
\emph{\underline{local} $(i, t)$-status} of the leaf $\ell_x$,
where the $(i, t)$-bit in the bitmap 
is set to $1$ if and only if
$x$ and $v$ are both $(i,t)$-nodes 
but the fundamental $\H$-shortcut $\scut v x$
is \emph{not} an $(i,t)$-shortcut.
In a quiescent state, this only
occurs when $v$ is an $(i,t)$-branching node or $(i,t)$-root
and $x$ is an $(i,t)$-node.\footnote{Notice 
that if $x$ is an $(i,t)$-branching node but $v$ is not,
then $x\in \H$ has $(i,t)$-status but $\ell_x \in \L(v)$ does 
not have \emph{local} $(i,t)$-status.
This is the main reason for notationally 
distinguishing $x$ from $\ell_x$.}  
However, in the middle
of a \Delete{} operation we may temporarily 
uncover and remove a fundamental $(i,t)$-shortcut 
$\scut v x$,
which can cause $v,x$ to temporarily become $(i,t)$-nodes
and $\ell_x\in \L(v)$ to temporarily become a \emph{local}
$(i,t)$-node.

\paragraph{Local $(i, t)$-trees.}
Consider a local tree $\L(v)$.
The \emph{local $(i,t)$-nodes} comprise
all leaves of $\L(v)$ with local $(i,t)$-status,
as well as those 
internal nodes $z\in\L(v)$ 
satisfying at least one of the following.

\begin{itemize}\compactify
\item $z$ is the root of $\L(v)$, having at least one leaf-descendant 
with local $(i,t)$-status.
\item $z\in \L(v)$ is a bottom tree node, a buffer tree node, or a top tree node having a descendant with local $(i,t)$-status.  (Because of their dual membership,
middle tree roots and leaves are also included in this category.)
\item $z\in \L(v)$ is a middle tree node whose children both have descendants with $(i,t)$-status.  (It is a \emph{local $(i, t)$-branching} node.)
\item $z\in \L(v)$ is a child of a middle tree $(i, t)$-branching node.
\end{itemize}

A \emph{local $(i, t)$-tree} is defined in a similar fashion as the 
$(i, t)$-forest on $\H$; each
$z\in \L(v)$ maintains a bitmap indicating for which $(i,t)$-pairs it is a local $(i,t)$-node.
Whereas the $(i,t)$-forest can have arbitrary branching factor, 
every local $(i, t)$-tree is binary since $\L(v)$ is itself binary. 
Navigating from a local $(i,t)$-node $z$ to its child is straightforward
when $z$ is in a bottom, buffer, or top tree, since the $(i, t)$-bits are stored explicitly at every node in these trees, and these trees are binary.
\emph{Local} $(i,t)$-shortcuts are used for faster navigation in the middle layer; 
these are defined in Section~\ref{subsection:middle-trees}.
Each local $(i,\const{primary})$-tree
node in $\L(v)$
maintains an approximate 
$i$-counter.


In Sections~\ref{sect:bottom-buffer}--\ref{sect:toptrees} we describe the operations of the 
bottom/buffer layer, the middle layer, and the 
top layer in isolation.  In particular, Lemmas~\ref{lemma:buffer-tree-operations},
\ref{lemma:middle-tree-operations},
and 
\ref{lemma:top-tree-operations}
state the \emph{worst case} cost of operations,
without regard to side effects on other layers.
The interaction between the 
layers and the amortization of costs 
is addressed in 
Section~\ref{subsection:local-tree-operations},
Lemma~\ref{lemma:local-tree-operations}.

\subsection{Bottom Trees and the Buffer Tree}\label{sect:bottom-buffer}

The algorithm attaches new $\H$-node representatives only to the buffer tree, while deletions of $\H$-node representatives can take place in both buffer and bottom trees.
The buffer tree can be regarded as a bottom tree under construction.

Each buffer tree and bottom tree has at most $2\log^\alpha n$ local tree leaves, 
where $\alpha$ is a constant to be determined in Section~\ref{section:local-tree-cost-analysis}.
Whenever the buffer tree size exceeds $\log^\alpha n$, 
either from merging two $\H$-nodes or from inserting a new local tree leaf,
the buffer tree becomes mature and is converted to a bottom tree. 
The data structure adds this bottom tree into the bottom layer 
and creates a new empty buffer.  

The buffer and bottom trees are $O(\log\log n)$ height mergeable binary trees.
Each node stores a weight, a vector of approximate counters, 
pointers to its parent and children, and a bitmap
indicating for each $(i,t)$ pair, whether there is a 
local tree leaf in its subtree with local $(i,t)$-status.

\begin{lemma}\label{lemma:buffer-tree-operations}
The buffer tree and bottom trees support the following operations,
with the following worst case time complexities (in parentheses):
\begin{itemize}\compactify
\item Detach a buffer/bottom tree leaf \ $(O((\log\log n)^2))$.
\item Remove local $(i, t)$-status from a buffer/bottom tree leaf \ $(O(\log\log n))$.
\item Given an edge depth $i\in [1,\dmax]$, a buffer/bottom leaf $x$, and a value $q$, \emph{decrease} the approximate $i$-counter at $x$ to $q$ \ $(O(\log\log n))$.
\end{itemize}
In addition, the buffer tree supports the following operations:
\begin{itemize}\compactify
\item Attach a buffer tree leaf \ $(${}$O((\log\log n)^2)${}$)$.
\item Merge two buffer trees of two sibling $\H$-nodes \ $(${}$O((\log\log n)^2)${}$)$.
\item Add local $(i,t)$-status to a buffer tree leaf \ $(O(\log\log n))$.
\item Convert the buffer tree to a bottom tree $(O(1))$.
\item Given an edge depth $i\in [1,\dmax]$, a buffer leaf $x$,
and a value $q$, 
set the approximate 
$i$-counter at $v$ to be $q$ \ $(O(\log\log n))$.
\end{itemize}
\end{lemma}

\begin{proof}
\begin{sloppypar} A buffer tree is implemented by an off-the-shelf mergeable binary tree with
$O(\log\log n)$ worst case time for each attach, detach, and merge 
operation.\footnote{Note that all such off-the-shelf data structures are, in fact, binary \emph{search} trees, but we do not impose any total order on the leaves, nor do we require any operation analogous to binary \emph{search}.}
However, in order to support updates to the vector of
approximate counters, an $O(\log\log n)$ factor overhead 
is applied to each of the operations.
See Lemma~\ref{lemma:approximate-counters-addition}.
Hence the worst case time 
cost for each operation is $O((\log\log n)^2)$.
From these three operations, bottom trees are only subject to \emph{detach}.
Since we only require the height of a bottom tree to be $O(\log\log n)$,
no rebalancing is necessary after detaching a leaf.
In order to obtain correct $\text{rank}(x)$, each attach, detach, and merge also updates the weight of the given buffer/bottom tree root.\end{sloppypar}

To add $(i,t)$-status to a buffer tree leaf $x$, the data structure
traverses up the buffer tree and sets the $(i,t)$-bit to $1$ in all 
ancestors of $x$ in the buffer tree.
To remove $(i, t)$-status from a buffer/bottom leaf $x$, 
the data structure updates the $(i, t)$-bits at each ancestor of $x$.  If a leaf $x$ has local $(i,\const{primary})$ status, 
it carries an approximate $i$-counter.  Such counters can be 
increased or decreased in $O(\log\log n)$ time by updating all ancestors of $x$ in its buffer/bottom tree.
\end{proof}

\begin{remark}\label{remark:bottom-trees}
Observe that \emph{only} the buffer tree can acquire
new leaves, and only buffer tree nodes can 
\emph{gain} local $(i,t)$-status and 
\emph{increase} their approximate $i$-counters.
In particular, this implies that when a bottom tree leaf has to acquire a local $(i, t)$-status, the algorithm removes the leaf from the bottom tree, updates its status and re-inserts the leaf into the buffer tree.
\end{remark}

\subsection{Middle Trees}\label{subsection:middle-trees}

All bottom tree roots are middle tree leaves.
Middle trees respond to three types of updates at their 
leaves: a leaf losing $(i,t)$-status, decreasing its 
approximate $i$-counter, or decreasing its weight.
Middle trees are maintained as weight-balanced binary trees
satisfying Invariant~\ref{invariant:middlerank}.

\begin{invariant}\label{invariant:middlerank}
If $x$ is a middle tree leaf/bottom tree root it maintains $w(x)$ and
$\mathrm{rank}(x)=\floor{\log w(x)}$.
If $x$ is an internal
middle tree node
it maintains only
$\mathrm{rank}(x)$,
and if $x$ has children 
$x_L,x_R$ then
$\mathrm{rank}(x_L) = \mathrm{rank}(x_R) = \mathrm{rank}(x) - 1$.
\end{invariant}

The operations described in Lemma~\ref{lemma:middle-tree-operations}
specifically maintain Invariant~\ref{invariant:middlerank}.
As a consequence of Invariant~\ref{invariant:middlerank},
the path from any middle tree leaf $x_B$ (bottom tree root) to the corresponding middle tree root $x_M$ has length $\log(\f{w(x_M)}{w(x_B)})+O(1)$. 
This property is used to bound the number of local tree nodes traversed
when walking from any $\H$-node to its $\H$-parent $v$ via the local tree $\L(v)$.
In accordance with Invariant~\ref{invariant:middlerank},
two middle tree roots with the same rank may 
\emph{join}, and become
children of a new
middle node parent.

\paragraph{Local Shortcuts.}

Each of the middle trees maintains a \emph{local shortcut infrastructure} in much the same way that shortcuts are maintained in $\H$.
Let $u$ and $v$ be two nodes in the same middle tree such that $u$ is a proper ancestor of $v$. Then $\scut uv$ is an eligible local shortcut if and only if for every 
internal node $x$ on the path $P_{uv}$,
\[
\LSBIndex(\mathrm{rank}(x)+1) < \min\Big(\LSBIndex(\mathrm{rank}(u)+1),\; \LSBIndex(\mathrm{rank}(v)+1)\Big).
\]
Notice that the $\H$-shortcuts are defined from the depths of $\H$-nodes 
which increase along the path from an $\H$-root to
an $\H$-leaf.  In contrast, in middle trees the ranks of middle tree 
nodes decrease on the path from a middle tree root to a middle tree 
leaf.  The definition of \emph{power} is symmetric between $u$ and $v$,
so the increasing/decreasing direction here 
does not matter. Local shortcuts have the same properties as $\H$-shortcuts: they are non-crossing; all eligible local shortcuts
naturally form a poset $\succeq$, and the maximal elements (w.r.t.~$\succeq$) 
among shortcuts on a middle tree path
$P_{uv}$ form a path $\ShortCuts(u,v)$ with length $O(\log\log n)$.
A local shortcut with power $0$ is called a \emph{trivial shortcut},
which coincides with a middle tree edge from a parent to one of its children.

Invariant~\ref{invariant:local-it-trees}
is a local tree analogue of Invariant~\ref{invariant:it-shortcuts}.

\begin{invariant}\label{invariant:local-it-trees}
Let $u$ be a single-child local $(i, t)$-node and let $v$ be the local 
$(i, t)$-child of $u$. Then the local $(i, t)$-shortcuts on $P_{uv}$ that 
are stored by the data structure form a path connecting $u$ and $v$.
\end{invariant}

Lemma~\ref{lemma:middle-tree-operations} lists 
the middle tree operations.

\begin{lemma}\label{lemma:middle-tree-operations}
The data structure supports the following operations 
on a collection of middle trees, maintaining Invariants~\ref{invariant:middlerank}
and~\ref{invariant:local-it-trees},
with the following worst case 
time complexities (in parentheses):
\begin{itemize}\compactify
\item Reduce the weight of a middle tree leaf \ $(O(\log n\log\log n))$.
\item Remove $(i,t)$-status from a middle tree leaf \ $(O(\log n))$.
\item Given an edge depth $i\in [1, \dmax]$ and a middle tree leaf $x_B$, 
update the approximate $i$-counter at $x_B$
\ $(O(\log n))$.
\item Given a newly created bottom tree root, create a new middle tree leaf \ $(O(1))$.
\item Join two middle trees with the same rank \ $(O(\log\log n))$.
\end{itemize}
\end{lemma}

\begin{proof} We address each operation in turn.
\paragraph{Reducing ranks.}
When the weight of a middle tree leaf $x_B$ is reduced (because its bottom
tree suffered enough leaf deletions) it may cause a discrete
reduction in its rank, which violates Invariant~\ref{invariant:middlerank}.
If so, we destroy all middle tree nodes
that are strict ancestors of $x_B$.
We first uncover all local shortcuts touching the path from $x_B$
to its middle tree root $x_M$. 
This procedure is the same as the uncovering procedure described in Section~\ref{section:uncover-a-path}.
In order to avoid redundancy, we do not provide details here.
This costs $O(\log n\log\log n)$ time,
and increases the number of middle trees in the collection. 
(Each new middle tree root is inserted into the top tree.)

\paragraph{Removing $(i, t)$-status.}
Similar to the $(i, t)$-forests, in the local $(i, t)$-tree the middle tree edges between a local $(i, t)$-branching node $x$ and its $(i, t)$-children are \emph{not} considered 
to be trivial $(i, t)$-shortcuts.
To remove $(i,t)$-status from a bottom tree root/middle tree leaf $x_B$, 
we follow local upward $(i, t)$-shortcuts to find the one-child $(i, t)$-node ancestor 
$x'$ of $x_B$.  If $x'=x_M$ is the middle tree root of $x_B$ then we remove
$(i,t)$-status from $x_M$ (triggering an update to the top tree; see Lemma~\ref{lemma:top-tree-operations}).
Otherwise, the parent of $x'$, $x''$ is an $(i,t)$-branching node.
We remove $(i,t)$-status from $x'$ and all shortcuts from $x_B$ to $x'$,
then add a trivial $(i,t)$-shortcut from $x''$ to the sibling of $x'$.
This may cause $x''$ and/or the sibling of $x'$ to lose $(i,t)$-status.
Since the middle trees are weight balanced, removing $(i, t)$-status from a middle tree leaf costs worst case $O(\log n)$ time.

\paragraph{Update an approximate $i$-counter.} 
If the approximate $i$-counter at $x_B$ changes it invalidates the approximate $i$-counters at all ancestors on the path from $x_B$ to its middle tree root $x_M$.
Each can be updated
in $O(1)$ time
(Lemma~\ref{lemma:approximate-counters-addition}),
for a total of 
$O(\log n)$ time.

\paragraph{Create a new middle tree leaf.}
The buffer tree root maintains its weight
and approximate $i$-counters.  Thus, when
the buffer is converted to a bottom tree,
its root (the new middle tree leaf) can
be inserted into the middle tree collection
in $O(1)$ time.
(As a new middle tree root, it is also inserted 
as a leaf in the top tree;
this is accounted for in Lemma~\ref{lemma:top-tree-operations}.)

\paragraph{Joining middle trees.}
To join roots $x_L,x_R$, we create a new middle
tree parent $x$ and compute its approximate
$i$-counters in $O(\log\log n)$ time
(Lemma~\ref{lemma:approximate-counters-addition})
by adding the vectors at $x_L,x_R$.
We set the bitmap of $x$ to be the 
bitwise OR of bitmaps stored in $x_L$ and $x_R$.  
In order to maintain Invariant~\ref{invariant:local-it-trees},
the data structure adds trivial $(i,t)$-shortcuts whenever $x$ 
has an $(i, t)$-bit set to $1$ and
exactly one of $x_L$ or $x_R$ has its $(i,t)$-bit set to $1$. 
This is done in $O(1)$ time using bitwise operations.
\end{proof}

\subsection{Top Trees}\label{sect:toptrees}

The top tree is an $O(\log\log n)$-height mergeable binary tree.
All middle tree roots are top tree leaves.
As a consequence of the middle tree reduction procedure described below, 
each top tree has at most $4\log n$ top tree leaves. Each top tree node 
$x$ maintains pointers to its parent and children,
approximate counters, and a bitmap of $(i, t)$ pairs indicating whether a 
local tree leaf with $(i, t)$-status appears in the subtree of $x$.


Whenever we invoke the Middle Tree Reduction procedure, the entire top tree is rebuilt.

\paragraph{Middle Tree Reduction.}

There are at most $\log n$ possible ranks for a middle tree node.
If there are at least $2\log n$ middle trees in a local tree, 
then the data structure invokes the
\emph{middle tree reduction} procedure:
(1) destroy the top tree,
(2) repeatedly take two middle tree roots with the same rank, 
and \emph{join} the corresponding middle trees,
then (3) rebuild the top tree on the remaining middle tree roots.
The size of the top tree can be as large as $4\log n$ immediately 
after merging the top trees of two sibling $\H$-nodes.

\begin{lemma}\label{lemma:top-tree-operations}
The following operations are supported on the top trees,
with the following worst case time complexities (in parentheses):
\begin{itemize}\compactify
\item Insert a middle tree root into the top tree \ $(O((\log\log n)^2))$.
\item Remove a middle tree root from the top tree $(O((\log\log n)^2))$.
\item Merge the top trees of two local trees $(O((\log\log n)^2))$.
\item Given the list of all middle tree roots that are leaves of the top tree, 
    perform a middle tree reduction and rebuild the top tree $(O(\log n\log\log n))$.
\item Update approximate counters along the path from the given top 
    tree leaf $x_M$ to the top tree root $x_T$\ $(O((\log\log n)^2))$.
\item Remove $(i, t)$-status of a given middle tree root $(O(\log\log n))$.
\end{itemize}
\end{lemma}

\begin{proof}
The top tree implements leaf-insertion, leaf-deletion, and merging the two top trees in $O((\log\log n)^2)$ time.
Rebuilding the top tree costs time proportional to the number of middle trees (which is $O(\log n)$), 
multiplied by $O(\log\log n)$ for updating approximate counters at each node.
\end{proof}

\subsection{Maintaining Precision when Sampling}\label{sect:precision-sampling}

Recall from Invariant~\ref{invariant:approximate-counters}
that $H(x^j)$ was defined as the maximum possible height of any arithmetic
formula tree (summing up approximate counters) with $x^j \in \hat{V}_j$ at the root.
We define a similar function for nodes inside local trees.
If $v\in \L(x^j)$, define $H_{\ell}(v)$ as:
\begin{align*}
H_{\ell}(v) &= (\dmax-j-1)\cdot O(\log\log n) + \lfloor\log (w(v))\rfloor+ h_{\mathrm{bot/top}}(v),
\intertext{where $h_{\mathrm{bot/top}}(v) = O(\log\log n)$ is precisely the maximum
number of top, bottom, and buffer trees nodes on a path from $v$ to a leaf of $\L(x^j)$.
With this definition, it is straightforward to see that when $v_L,v_R$ are the children of $v$, that}
H_{\ell}(v) &= \max(H_{\ell}(v_L), H_{\ell}(v_R)) + 1.
\end{align*}

We first prove that all nodes in a local tree have the correct precision in terms of $H_{\ell}(v)$.

\paragraph{Maintaining Invariant~\ref{invariant:approximate-counters}.}
Invariant~\ref{invariant:approximate-counters} constrains
the accuracy of approximate $i$-counters in terms of $H(\cdot)$.
We prove that Invariant~\ref{invariant:approximate-counters} is maintained, by analyzing the accuracy
of approximate $i$-counters inside the local trees in terms of $H_\ell(\cdot)$.

Fix an edge depth $i$ and a 
local $(i, \const{primary})$-branching node $x\in \H$. 
Assume, inductively, that
every local $(i, \const{primary})$-leaf $\ell_y$ in $\L(x)$ representing the $(i, \const{primary})$-child $y$ of $x$ satisfies 
Invariant~\ref{invariant:approximate-counters}
and $\hat C_i(\ell_y) = \hat C_i(y)$.  
We now prove that Invariant~\ref{invariant:approximate-counters} 
is satisfied at $x$ as well.
Fix a local $(i, \const{primary})$-branching 
node $v\in \L(x)$, and let $v_L,v_R$ be its $(i, \const{primary})$-children,
so $\hat C_i(v) = \hat C_i(v_L)\boxplus \hat C_i(v_R)$.
By induction on $H_{\ell}(v)$, \begin{align*}
\hat C_i(v) &\ge \left(1 -\log^{-\beta} n\right) (\hat C_i(v_L) + \hat C_i(v_R))\\
&\ge \left(1 - \log^{-\beta} n\right)^{\max(H_{\ell}(v_L), H_{\ell}(v_R))+1} (C_i(v_L) + C_i(v_R)) \\
&\ge \left(1 - \log^{-\beta} n\right)^{H_{\ell}(v)} C_i(v).
\end{align*}
On the other hand, by the definition of $\boxplus$ and the inductive hypothesis,
$
\hat C_i(v) \le \hat C_i(v_L) + \hat C_i(v_R)
\le C_i(v_L) + C_i(v_R) = C_i(v)
$.
In addition, for any single-child \emph{local} $(i, \const{primary})$-node $u$, 
the approximate $i$-counter $\hat C_i(u)$ is identical to
the approximate $i$-counter value from its local $(i, \const{primary})$-child $v$.
Since $H_{\ell}(u) \ge H_{\ell}(v)$, the precision requirement still holds.

Let $z$ be the root of $\L(x)$.  
Then $H_{\ell}(z) \le H(x)$ 
(provided the leading constants hidden by the
$O(\log\log n)$ factors in
the definitions of $H_\ell$ and $H$ are set correctly)
and 
Invariant~\ref{invariant:approximate-counters} holds for $x$ as well.

\subsubsection{Sample an \texorpdfstring{$(i, \const{primary})$}{(i,primary)}-child}

This section 
shows that an $(i,\const{primary})$-child
can be efficiently sampled approximately 
proportional to its approximate $i$-counter.

\begin{lemma}\label{lemma:local-tree-sampling}
Given an $(i,\const{primary})$-branching $\H$-node $u^{j-1}$,
we can sample an
$(i,\const{primary})$-child $u^j$ with probability
at most
\[
\frac{\hat C_i(u^j)}{\hat C_i(u^{j-1})}\cdot (1-\log^{-\beta} n)^{-(H(u^{j-1})-H(u^j))}
\]
in time
$O(H(u^{j-1})-H(u^j))$.
Recall that $\beta=2$ is constant.
\end{lemma}

The data structure begins at the root of $\L(u^{j-1})$, which is a local $(i,\const{primary})$ node,
and walks down to a descendant leaf in $\L(u^{j-1})$ as follows.
If we are at a local $(i, \const{primary})$-branching node $x$,
let $x_L$ and $x_R$ be its local $(i, \const{primary})$-children.
We randomly choose a child with probability proportional to $\hat C(x_L)$ and $\hat C(x_R)$, respectively,
and navigate
downward using local $(i, \const{primary})$-shortcuts to find the next local
$(i, \const{primary})$-branching child.  The process terminates when we reach a local leaf $\ell_{u^j}$ (representing $u^j$)
with local $(i, \const{primary})$-status.

Let $x_0$ be the root of $\L(u^{j-1})$, and the sequence $x_1, x_2, \ldots, x_{k}$ be all
local $(i, \const{primary})$-branching nodes which are on the path between $x_0$ and $x_{k+1}=\ell_{u^j}$.
For all $t\in[0,k]$, let $x'_t$ and $x''_t$ be the two local $(i, \const{primary})$-children of
$x_t$, with $x'_t$ being the ancestor of $x_{t+1}$.\footnote{In the case of $t=0$, if the root is not
a local $(i,\const{primary})$-branching node then we take $\hat C_i(t_0'')$ to be zero.}
Then we have for all $t \in [0,k]$,  $\hat C_i(x'_{t})= \hat C_i(x_{t+1})$,
and the probability that a particular $(i, \const{primary})$-child $u^j$ is sampled is at most
\begin{align*}
\prod_{t=0}^{k} \frac{\hat C_i(x'_{t})}{\hat C_i(x'_{t}) + \hat C_i(x''_{t})}
   &\le \prod_{t=0}^k \left[\frac{\hat C_i(x'_t)}{\hat C_i(x'_t) \boxplus \hat C_i(x''_t)} (1-\log^{-\beta} n)^{-1}\right]
\\
   &= \prod_{t=0}^k \left[\frac{\hat C_i(x_{t+1})}{\hat C_i(x_t)} (1-\log^{-\beta} n)^{-1}\right]
\\
   &= \f{\hat C_i(x_{k+1})}{\hat C_i(x_0)} (1-\log^{-\beta} n)^{-(k+1)}\\
   &\le  \f{\hat C_i(u^j)}{\hat C_i(u^{j-1})} (1-\log^{-\beta} n)^{-(H(u^{j-1}) - H(u^j))}.
\end{align*}

\subsection{Local Tree Operations}\label{subsection:local-tree-operations}

Lemmas~\ref{lemma:buffer-tree-operations}, \ref{lemma:middle-tree-operations},
and \ref{lemma:top-tree-operations} established worst case bounds on the elementary 
operations inside buffer, bottom, middle, and top trees.  
Lemma~\ref{lemma:local-tree-operations} 
lists the operations supported by the local tree as a whole,
and analyzes their \emph{amortized} time costs.

\begin{lemma}\label{lemma:local-tree-operations}\label{section:local-tree-cost-analysis}
   There exists a data structure that supports the following operations between an $\H$-node $v$ and its $\H$-children, with the following amortized time complexities (in parentheses):
   \begin{itemize}\compactify
   \item Attach a new $\H$-child $x$ to $v$ $\left(O((\log\log n)^2)\right)$.
   \item Detach an $\H$-child $x$ of $v$ $\left(O((\log\log n)^2)\right)$.
   \item Let $S$ be a set of $\H$-children of $v$. 
   Merge $S$ into a single node $x'$, which is a new 
   $\H$-child of $v$.
                                    $\left(O(|S|(\log\log n)^2)\right)$.
   \item Given a non-root $\H$-node $x$, return its $\H$-parent $v$  $\left(O(H(v) - H(x))\right)$.
   \item Given an $\H$-node $v$, enumerate all $\H$-children of 
        $v$ with $(i,t)$-status $(${}$O(\log\log n)$ per child$)$
        or decide if $v$ has a unique $(i,t)$-child $\left(O(\log\log n)\right)$.
   \item Given an $\H$-node $x$ and a bit vector $b$, add local $(i,t)$-status to the local tree leaf $\ell_x$ for all $(i, t)$ pairs indicated by $b$ $(O((\log\log n)^2))$.
   \item Given an $\H$-node $x$ and a bit vector $b$, delete local $(i,t)$-status to the local tree leaf $\ell_x$ for all $(i, t)$ pairs indicated by $b$ $(O(\log\log n))$.
   \item Given an $(i, \const{primary})$-branching node $v$, sample 
        an $(i, \const{primary})$-child $x$ with probability at most
   \[
   \f{\hat C_i(x)}{\hat C_i(v)}\cdot (1-\log ^{-2} n))^{-(H(v)-H(x))}.  \hspace{2\ThCScm} \mbox{(Time: $O(H(v)-H(x))$)}
   \]
   \item Increase the $i$-counter of an $\H$-child $x$ of $v$. $\left(O((\log\log n)^2)\right)$.
   \item Decrease the $i$-counter of an $\H$-child $x$ of $v$. $\left(O(\log\log n)\right)$.
   \end{itemize}
   \end{lemma}

\begin{proof}
We will address these operations one by one.
The \emph{sampling} operation 
was already established in 
Lemma~\ref{lemma:local-tree-sampling}, Section~\ref{sect:precision-sampling}.
We first describe the worst case cost of operations, and at the end of the proof we analyze the amortized cost.

\paragraph{Attach a new $\H$-child $x$.} The local tree leaf $\ell_x$ is created 
and inserted into the buffer tree of $\L(v)$.
By Lemma~\ref{lemma:buffer-tree-operations} the worst case cost of this operation 
is $O((\log\log n)^2)$.
If the buffer tree is full, the algorithm converts the buffer tree into a bottom tree which costs $O(1)$ by Lemma~\ref{lemma:buffer-tree-operations},
then creates a middle tree leaf which costs $O(1)$ time by Lemma~\ref{lemma:middle-tree-operations},
and possibly rebuilds the top tree which costs $O((\log\log n)^2)$ time by Lemma~\ref{lemma:top-tree-operations}.

\paragraph{Detach an $\H$-child $x$.} 
The local tree representitive $\ell_x$ is first removed from either the corresponding buffer tree or bottom tree, costing   $O((\log\log n)^2)$ time by Lemma~\ref{lemma:buffer-tree-operations}.
In the case where the corresponding buffer/bottom tree root loses its local $(i, t)$-status, or in the case where the approximate $i$-counters are reduced, the entire ancestor path is updated in $O(\log n)$ time by Lemmas~\ref{lemma:middle-tree-operations} and \ref{lemma:top-tree-operations}.
In the case where the rank of the corresponding buffer/bottom tree root is reduced (costing $O(\log n\log\log n)$ time by Lemma~\ref{lemma:middle-tree-operations}), the middle tree reduction may be then invoked, costing $O(\log n\log\log n)$ time by Lemma~\ref{lemma:top-tree-operations}.
Notice that these $\log n$ worst case terms are amortized away at the end of this proof.

\paragraph{Merge sibling $\H$-nodes.}
For each node $x\in S$, we detach the representative $\ell_x$ in worst case $O((\log\log n)^2)$ time by Lemma~\ref{lemma:buffer-tree-operations},
and then merge the local trees of $S$-nodes 
in pairs until there is only one node left.
To merge local trees we first merge their buffer trees
(costing $O((\log\log n)^2)$ time by Lemma~\ref{lemma:buffer-tree-operations}),
then merge their top trees (costing $O((\log\log n)^2)$ by Lemma~\ref{lemma:top-tree-operations}).
Then, if the merged buffer tree is full, make it a bottom tree
(costing $O(1)$ by Lemma~\ref{lemma:middle-tree-operations}).
Finally, if the top tree is full,
call the middle tree reduction procedure
(costing $O(\log n\log\log n)$ time by Lemma~\ref{lemma:top-tree-operations}).
Let $x'$ be the node resulting from merging $S$. 
Its representative $\ell_{x'}$ is created, 
having weight that is the sum of the weights of the $S$-nodes, and 
reattached to the buffer tree of $\L(v)$,
in $O((\log\log n)^2)$ time.  

\paragraph{Return the $\H$-parent.} Let $x\in \H$ be a non-root $\H$-node.
We find the local representative $\ell_x\in \L(v)$, then walk up to the 
root of $\L(v)$ and return ``$v$.''  The number of buffer, bottom, and top
tree nodes traversed is $O(\log\log n)$ and the number of middle tree nodes
traversed is $\log(\f{w(v)}{w(x)})+O(1)$.  By the definition of $H(\cdot)$,
this is bounded by $H(v)-H(x)$.

\paragraph{Enumerate all local tree leaves with local $(i, t)$-status.} 
We perform a depth-first search from the local tree root. 
When the search encounters a top tree, a bottom tree, 
or a buffer tree node, the bitmaps in its children indicate whether 
the child contains a local tree leaf with an $(i, t)$-status or not. 
When the search encounters a middle tree node $x$, we examine
$\Dsc[\Dptr[i,t]]$ to see whether there is a downward 
local $(i, t)$-shortcut 
leaving $x$ or not. If there is no downward local $(i, t)$-shortcut leaving $x$,
then $x$ is a local $(i, t)$-branching node and 
the search proceeds recursively 
on both children. Otherwise, the search navigates downward from a local $(i, t)$-non-branching node $x$ to its highest descendant $(i, t)$-node $x'$. 
The same navigation algorithm described in Section~\ref{section:the-H-shortcut-data-structure} is performed so that after the navigation all $(i, t)$-shortcuts on the path $P_{uv}$ are exactly local shortcuts 
in $\ShortCuts(x,x')$.  (The cost of adding these shortcuts
inside a local tree is amortized differently than 
how adding $\H$-shortcuts are amortized; see below.)
All local tree leaves with $(i,t)$-status are enumerated 
in $O(\log\log n)$ amortized time per leaf.

To test whether there is a \emph{unique} leaf with $(i, t)$-status, 
we navigate downward from the root $z$ of $\L(v)$, 
following local $(i,t)$-shortcuts
until reaching either a local $(i,t)$-leaf $x$ (necessarily unique)
or a local $(i,t)$-branching node $x$ (indicating non-uniqueness).
We then cover local $(i,t)$-shortcuts on the path from $z$ to $x$
as long as it is possible. As shown below, the
amortized cost
of this operation is $O(\log\log n)$.

\paragraph{Add local $(i, t)$-status to a local tree leaf.}
Recall that the \emph{only} leaves that
may gain local $(i,t)$-status are buffer tree leaves (Remark~\ref{remark:bottom-trees}).  
Let $\ell_x$ be the local tree leaf gaining $(i,t)$-status.
If $\ell_x$ resides in a bottom tree we detach it, reattach 
it to the buffer tree, and add $(i,t)$-status there.
From the description above (the first two operations listed on Lemma~\ref{lemma:local-tree-operations}), the time cost is amortized $O((\log\log n)^2)$ due to the potential detach/attach operation.

\paragraph{Delete local $(i, t)$-status from a local tree leaf.}
The algorithm first removes the local $(i, t)$-status from the local tree representative $\ell_x$, costing $O(\log\log n)$ time by Lemma~\ref{lemma:buffer-tree-operations}. If the corresponding bottom tree root loses some local $(i, t)$-status, the algorithm removes local $(i, t)$-status from the corresponding middle tree leaf, costing $O(\log n)$ time by Lemma~\ref{lemma:middle-tree-operations}. The $\log n$ worst case time cost will be amortized as described below.

\paragraph{Increase the approximate $i$-counter of an $\H$-child $x$ of $v$.}
Let $\ell_x$ be the local tree leaf that represents $x$. The algorithm detaches $\ell_x$, changes the $i$-counter value and then attaches $\ell_x$ to the buffer tree. The operations costs $O((\log\log n)^2)$ time from the first two operations listed on Lemma~\ref{lemma:local-tree-operations}.

\paragraph{Decrease the approximate $i$-counter of an $\H$-child $x$ of $v$.}
The algorithm sets the approximate $i$-counter at $x$ to the new value, costing $O(\log\log n)$ time by Lemma~\ref{lemma:buffer-tree-operations}. If $\ell_x$ is in a bottom tree and the corresponding bottom tree root has its approximate $i$-counter value changed, the algorithm invokes Lemma~\ref{lemma:middle-tree-operations} and updates the approximate $i$-counter at the corresponding middle tree leaf, costing $O(\log n)$ worst case time and again can be amortized away by the description below.

\paragraph{Amortized Cost Analysis.}
We use a credit system.  
Every buffer tree leaf carries $\Theta(1)$ credits
and every middle tree root carries $\Theta(\log\log n)$ credits.
Suppose the buffer tree matures and becomes a bottom tree, say with root $x_B$.  
At this moment the tree has $\Theta(\log^\alpha n)$ credits, which will pay for all future costs associated
with updating the middle and top tree ancestors of $x_B$.
The following three types of events change
the information stored at $x_B$.
\begin{enumerate}\compactify
\item $x_B$ changes rank. Since the bottom tree is only subject to \emph{detach} operations (see Remark~\ref{remark:bottom-trees}), 
its weight is non-increasing.  
Therefore, this happens at most $\log n$ times.
\item $x_B$ loses $(i,t)$-status.  
It can never regain $(i,t)$-status (Remark~\ref{remark:bottom-trees}), 
so this happens 
at most $3\dmax = O(\log n)$ times.
\item $x_B$'s approximate $i$-counter changes.  
The approximate counters are non-increasing, 
and each such counter can take on $O(\log^{\beta+1} n)$
different 
values (Section~\ref{section:the-sampling-procedure}).  
Since there are $\log n$ possible values for $i$, 
the total number of counter changes
is $O(\log^{\beta+2} n)$.
\end{enumerate}
Each of the above events requires that we update or delete the entire path form $x_B$ 
to the
local tree root, which can have length $\Theta(\log n)$.
Events of type (1) take $O(\log n\log\log n)$ time to destroy the path
and reinsert new middle tree roots into the top tree, 
each with $O(\log\log n)$ credits.
Events of type (2) and (3) take $O(\log n)$ time to update the $(i,t)$-status
or approximate $i$-counters of all ancestors of $x_B$.
Since $\beta=2$, the total cost for events of type (3) is the bottleneck.  They take
$O(\log^{\beta+3} n)$ time over the life of the bottom tree.
We set $\alpha \ge \beta + 3=5$, so the credits of a bottom tree suffice 
to pay for all costs incurred over the lifetime of the bottom tree.

A middle tree reduction procedure is invoked if the leaf set $S$ of the top tree has
size $|S|\ge 2\log n$.  Thus, it begins with at least $2\log n\cdot O(\log\log n)$
credits and ends with at most $\log n \cdot O(\log\log n)$ credits, which
pays for rebuilding the top tree in $O(\log n\log\log n)$ time (Lemma~\ref{lemma:top-tree-operations}).

The number of shortcuts removed is bounded by the number created,
so it suffices to account for the cost of creating shortcuts.
Local shortcuts are created in two ways:
(i) in response to the creation of a middle tree node (joining two middle trees),
 and
(ii) lazy covering.
The cost of case (i) is ultimately paid for by the deletion of that middle tree node,
which in turn is paid for by the bottom tree that triggered the deletion.
The cost of case (ii) is attributed to the removal of $(i, t)$-status
at some corresponding middle tree leaf with an $(i, t)$-status,
which is accounted for in the cost of type (2) events.
\end{proof}

\section{Loose Ends}
\label{section:general-operations}

Some of the operations on the hierarchy $\H$ required the 
definition of $(i,t)$-forests (Section~\ref{section:shortcut-infrastructure-on-H}) 
and local trees (Section~\ref{section:local-trees}) and could not
be described until now.  In Section~\ref{section:the-batch-sampling-test}
we analyze the cost of searching for a replacement edge using the
two-stage batch sampling test sketched in Section~\ref{section:finding-a-replacement-edge}.
In Section~\ref{section:proof-lemma-induced-forests} 
we explain how to 
maintain $(i,t)$-forests (Invariant~\ref{invariant:it-shortcuts}),
and in particular, how to efficiently merge two such forests
when doing batch promotions/upgrades.
In Section~\ref{section:proof-of-lemma-counters} 
we prove Lemma~\ref{lemma-counters} concerning approximate counters,
and show that
Invariant~\ref{invariant:approximate-counters} is maintained.

\subsection{The Batch Sampling Test}
\label{section:the-batch-sampling-test}

Recall from the deletion algorithm of Section~\ref{section:establishing-two-components} that $\bf u^i$ is
the new $\H$-node resulting from merging a set of siblings.
In this section we show how the data structure performs the batch sampling test among $i$-primary endpoints
touching $\bf u^i$.
Let $p$ and $s$ be the number of $i$-primary and $i$-secondary endpoints touching $\bf u^i$,
and let $\hat p = \hat C_i({\bf u^i})$ be a $(1+o(1))$-approximation of $p$.
(Retrieving $\hat p$ is Operation~\ref{opr:counter}
from Lemma~\ref{lemma-main}.)

\paragraph{Single Sample Test.}
To $(1+o(1))$-uniformly sample one $i$-primary endpoint touching $\bf u^i$, the data structure sets
$x={\bf u^i}$ and iteratively performs the following step.
Base case: If $x$ is an $(i, \const{primary})$-leaf, then return an $i$-primary endpoint at $x$ uniformly at random.
General case: If $x$ is an $(i, \const{primary})$-branching node,
then use $\L(x)$ to sample an $(i, \const{primary})$-child $x'$ of $x$
with probability at most $\f{\hat C_i(x')}{\hat C_i(x)}(1-\log^{-\beta} n)^{-(H(x)-H(x'))}$ (Lemma \ref{lemma:local-tree-sampling}). 
If $x'$ is an $(i,\const{primary})$-branching node or leaf, we set $x=x'$ and repeat.
Otherwise, we repeatedly follow the
$(i, \const{primary})$-shortcuts leaving $x'$ to its 
$(i, \const{primary})$-child $x''$, set $x=x''$, and repeat (Lemma \ref{lemma:lazy-cover-operation}).

Notice that with accurate counters this procedure picks a perfectly uniformly random $i$-primary endpoint.
Let $\Endpoint{x}{\{x, y\}}$ be the sampled endpoint and
$x_0={\bf u^i}, x_1, \ldots, x_k = x$ be the sequence of $(i, \const{primary})$-branching nodes
on the path in $\H$ from ${\bf u^i}$ to $x$.
Then the probability that $\Endpoint{x}{\{x, y\}}$ is sampled is at most
\begin{align*}
&
\frac{1}{C_i(x)} \prod_{j=0}^{k-1} \left[\frac{\hat C_i(x_{j+1})}{\hat C_i(x_{j})}(1-\log^{-\beta} n)^{-(H(x_{j}) - H(x_{j+1}))}\right]
\\
&= \frac{1}{C_i(x)} \left[ \prod_{j=0}^{k-1} \frac{\hat C_i(x_{j+1})}{\hat C_i(x_{j})} \right]
(1 - \log^{-\beta} n)^{-H(x_0)}												&\\
&= \frac{1}{C_i(x)} \frac{\hat C_i(x)}{\hat C_i({\bf u^i})} (1 - \log^{-\beta} n)^{-H(x_0)}		 \\
&= (1-o(1))\frac{1}{\hat C_i({\bf u^i})} (1 - \log^{-\beta} n)^{-O(\log n \log\log n)}			& (H(x_0) = O(\log n\log\log n))\\
&\le (1+o(1))\frac{1}{C_i({\bf u^i})}.											& (\beta = 2)
\end{align*}

The $1/C_i(x)$ factor reflects the fact that once we reach the $\H$-leaf $x$, an endpoint touching $x$
is selected (exactly) uniformly at random, without any approximation.
To check whether $\{x, y\}$ is a replacement edge or not, it suffices to check whether $y^i = {\bf u^i}$.
This can be accomplished by starting from $y$ and repeatedly accessing the $\H$-parent until $y^i$ is reached.
Using local trees, the cost of computing $\H$-parents along a path telescopes to $H(y^i) = O(\log n\log\log n)$.

\paragraph{The Preprocessing Method.}
Another way to sample $i$-primary endpoints is to first enumerate \ul{all} $i$-primary endpoints and \ul{all}
$i$-secondary endpoints touching $\bf u^i$ in $O((p+s)\log\log n)$ time, mark all enumerated endpoints and store
all $i$-primary endpoints in an array.
Then the data structure samples an $i$-primary endpoint uniformly at random from all enumerated $i$-primary
endpoints and checks whether the other endpoint is marked in $O(1)$ time.

\paragraph{Batch Sampling Test on $k$ Samples.}
The data structure runs the two sampling methods in parallel and halts when the first finishes.
Thus, the time to sample $k$ $(i,\const{primary})$-endpoints is
\[
O\left(\min\Big\{(p + s)\log\log n + k, \; k\log n\log\log n\Big\}\right).
\]

\subsubsection{Cost Analysis for Sampling Procedure}\label{section:cost-analysis-sampling-procedure}

As described in Section~\ref{section:finding-a-replacement-edge}, 
the sampling procedure either returns a replacement edge, or invokes the enumeration procedure.
Roughly speaking, if no replacement edge is found, 
the cost is charged to
either upgrades of $(i, \const{secondary})$ endpoints or promotions to $(i, \const{primary})$ endpoints.
If any replacement edge is found, the data structure is willing to pay $O(\log n (\log\log n)^2)$ cost because this happens at most once per $\Delete{}$ operation.

If the enumeration procedure is invoked, the data structure upgrades all enumerated $i$-secondary endpoints touching
$\bf u^i$ to $i$-primary endpoints,
and then all $i$-primary endpoints 
touching $\bf u^i$ associated with non-replacement edges are promoted to
$(i+1)$-secondary endpoints.
The first batch sampling test, when $k=O(\log\log \hat{p})=O(\log\log p)$, costs
\begin{align*}
T_1 &= O(\min((p+s)\log\log n,\; \log\log p\cdot \log n\log\log n)).\\
\intertext{The second batch sampling test ($k=O(\log p)$), if invoked, costs}
T_2 &= O(\min((p+s)\log\log n,\; \log p\cdot \log n\log\log n)).\\
\intertext{The enumeration procedure, if invoked, costs}
T_E &= O((p+s)(\log\log n)^2).
\end{align*}

Let $\rho$ be the fraction of $i$-primary endpoints touching $\bf u^i$ associated with replacement edges
before the execution of the sampling procedure.
The rest of the analysis is separated into two cases:
\begin{description}\compactify
\item[Case 1.] If $\rho \ge 3/4$, the probability that the first batch sampling test returns with a replacement edge is at least
$1-(1/4+o(1))^{O(\log\log p)} > 1 - 1/\log p$.\footnote{It is $1/4+o(1)$ because the sampling procedure is only $(1+o(1))$-approximate.}
The second batch sampling test, if invoked, returns a replacement edge if at least half the $O(\log p)$ endpoints sampled belong to replacement
edges.  By a standard Chernoff bound, the probability that the second batch \emph{fails} to return a replacement edge and halt
is $\exp(-\Omega(\log p)) < 1/p$.

The expected time cost is therefore
\[
T_1 + (1/\log p) T_2 + (1/p) T_E = O\left(\paren{\log n + \frac{p + s}{p}}(\log \log n)^2\right) = O((\log n + s)(\log\log n)^2)
\]
We charge the \Delete{}
operation $O(\log n(\log\log n)^2)$, which covers the expected cost of the two batch sampling steps \underline{and}
the expected cost of dealing with \emph{primary} endpoints if the enumeration step is reached.
If the enumeration step is reached, endpoint upgrades pay for the $\Theta(s(\log\log n)^2)$ cost of dealing with secondary endpoints.

\item[Case 2.] Otherwise, $\rho < 3/4$.
If the enumeration procedure is ultimately invoked, a $1-\rho = \Omega(1)$ fraction of the $i$-primary endpoints
touching $\bf u^i$ belong to non-replacement edges, which are promoted to depth $i+1$, and all $s$ $i$-secondary endpoints are
upgraded
to either $i$-primary or $(i+1)$-secondary status.  In this case the time cost is
\[
T_1 + T_2 + T_E = O((p+s)(\log \log n)^2),
\]
which is charged to the promoted edges/upgraded endpoints.
We need to prove that the probability of terminating after the second batch sampling test is sufficiently small.
If $\rho \ge 1/4$ then the probability of the first batch sampling test \emph{not} returning a replacement edge
is at most $(3/4 + o(1))^{O(\log\log p)} < 1/\log p$.  In this case the expected cost is
\[
T_1 + (1/\log p)T_2 = O(\log n(\log \log n)^2).
\]
If $\rho < 1/4$ then, by a Chernoff bound, the probability that at least half the sampled endpoints belong to replacement edges
is $\exp(-\Omega(\log p)) < 1/p$.  Therefore the expected cost when the enumeration procedure is not invoked with $\rho < 1/4$ is at most
\[
(1/p)(T_1 + T_2) = O(\log n \log \log n),
\]
which is charged to the \Delete{} operation.
\end{description}

\subsection{Maintaining \texorpdfstring{$(i, t)$}{(i,t)}-Forests}
\label{section:proof-lemma-induced-forests}

Lemma~\ref{lemma-induced-forests} summarizes the operations on $(i,t)$-forests which are implemented via the shortcut infrastructure and local trees, together with their corresponding time cost.

\begin{lemma}\label{lemma-induced-forests}
There exists a data structure on $\H$ supporting the following operations with amortized time (in parenthesis):

\begin{itemize}\compactify
\item Given an $\H$-leaf $x$ and an $(i, t)$ pair, designate $x$ an $(i, t)$-leaf $\left(O(\log n(\log\log n)^2)\right)$.
\item Given an $(i, t)$-leaf $x$, remove its $(i, t)$-status $\left(O(\log n(\log\log n)^2)\right)$.
\item Given an $(i, t)$-node $v$, return the $(i, t)$-parent of $v$ $\left(O(\log\log n)\right)$.
\item Given an $(i, t)$-node $v$, enumerate the $(i, t)$-children of $v$
$\left(O(1+ k \log\log n)\right.$ where $k$ is the number of enumerated $(i, t)$-children$)$.

\item Given an $(i, t)$-tree $\T$ rooted at $v$, an integer $i' \in [i,\dmax]$,
an endpoint type $t'$,  and two subsets of $(i, t)$-leaves $S^-$ and $S^+$
(these subsets need not be disjoint), update $\H$ so that all of the leaves in
$S^-$ lose their $(i, t)$-leaf status, and all leaves in $S^+$ gain $(i', t')$-leaf
status (if they did not have it before) $\left(O(|\T| (\log\log n)^2+1)\right)$.
\end{itemize}

Each operation assumes that, prior to the execution of the operation, Invariant~\ref{invariant:it-shortcuts} holds for all $\H$-nodes of depth $\ge i$, where $i$ is part of the input of the operation. Moreover, Invariant~\ref{invariant:it-shortcuts} is guaranteed to hold for all $\H$-node of depth $\ge i$ after each operation is completed.
\end{lemma}

The remainder of this section constitutes a proof of Lemma~\ref{lemma-induced-forests}.

\paragraph{Add $(i, t)$-status to an $\H$-leaf.}
Let $x$ be the $\H$-leaf. In order to identify the $(i, t)$-branching ancestor of $x$, the data structure climbs up $\H$ and finds the first $\H$-node $x'$ that is \emph{either} an $(i, t)$-node or has a downward $(i, t)$-shortcut $\scut {x'}{x''}$. If $x'$ is an $(i, t)$-branching node, then since the $\H$-child of $x'$ that is also an ancestor of $x$ is not an $(i, t)$-node, $x'$ is the $(i, t)$-branching ancestor of $x$. Otherwise, the data structure performs a binary search on the path $P_{x'x''}$ to find the $(i, t)$-branching ancestor as follows:

If $\scut{x'}{x''}$ is not a fundamental $(i, t)$-shortcut, the data structure uncovers $\scut {x'}{x''}$ into $\scut {x'}{y}$ and $\scut {y}{x''}$ and recurses to one of the two subpaths depending on whether $y$ is an ancestor of $x$ or not.
Otherwise, $\scut{x'}{x''}$ is fundamental, and in this case $x'$ is the branching node we are looking for.
Let $x'''$ be the ancestor of $x$ that is a child of $x'$.
We uncover $\scut{x'}{x''}$, give local $(i,t)$-status to $\ell_{x'''}$ and $\ell_{x''}$ in $\L(x')$,
and then cover all shortcuts on the path $P_{x''',x}$, using Lemma~\ref{lemma:uncover-a-path} (See Section~\ref{section:uncover-a-path}.)
The cost for walking up these local trees telescopes to $O(\log n\log\log n)$ by Lemma~\ref{lemma:local-tree-operations}.
Now suppose that $t=\const{primary}$.
For every $(i, t)$-branching node $y$ that is an ancestor of $x$,
the data structure updates
the approximate $i$-counter stored in $y$, using Lemma~\ref{lemma:local-tree-operations}.
Now, Invariant~\ref{invariant:it-shortcuts} is restored on all $\H$-nodes with depth $\ge i$ since all $(i, t)$-shortcuts between $x'''$ and $x$ form the path $P_{x''',x}$.
Since there are at most $\dmax=O(\log n)$ such $(i, t)$-branching nodes affected, the amortized cost is at most $O(\log n(\log\log n)^2)$.

\paragraph{Remove $(i, t)$-status from an $(i, t)$-leaf.}
Let $x$ be the $\H$-leaf.
The data structure navigates up from $x$ by upward $(i, t)$-shortcuts until it reaches a single-child $(i, t)$-node $q$.
The intermediate $(i, t)$-shortcuts are removed by setting their $(i, t)$-bits to $0$.

The data structure then removes the local $(i, t)$-status of the local tree leaf $\ell_q$ representing $q$.
If the $(i, t)$-parent $p$ of $q$ (which is also its $\H$-parent) now has only one $(i, t)$-child $q'$,
$p$ is no longer an $(i, t)$-branching node.
The data structure removes the $(i, t)$-status of $q'$, 
removes \emph{local} $(i,t)$-status of $\ell_{q'}$ in $\L(p)$ using Lemma~\ref{lemma:local-tree-operations},
removes the $(i, t)$-branching status of $p$, 
and covers the fundamental $(i, t)$-shortcut $\scut p{q'}$ using Lemma~\ref{lemma:uncover-a-path}.  
This may also cause $p$ to lose its $(i,t)$-status.

Notice that this operation is equivalent to first performing the lazy covering on the $(i, t)$-shortcuts from $x$ to its $(i, t)$-parent and then removing $x$. Hence, the time cost for removing $(i, t)$-status from $x$ is amortized $O((\log\log n)^2)$.  We can remove $(i,t)$-status from a group of leaves $S^-$ in
$O(|S^-|(\log\log n)^2)$
amortized time by repeating this procedure for every leaf.
Notice that Invariant~\ref{invariant:it-shortcuts} holds for all $\H$-nodes with depth $\ge i$ because fundamental $(i, t)$-shortcuts are covered when $\H$-nodes lose their $(i, t)$-branching status.

\paragraph{Enumerating $(i,t)$-children.} This is an operation of Lemma~\ref{lemma:local-tree-operations}.

\paragraph{Given an $(i, t)$-tree $\T$ and a set of leaves $S^+$ in $\T$, add $(i', t')$-status to the leaves in $S^+$.}
First of all, the data structure creates a ``dummy'' tree induced from the set of leaves $S^+$ and the root of $\T$,
by first copying the entire $(i, t)$-tree $\T$, enumerating all its leaves and removing all the leaves that do not belong to $S^+$.\footnote{This is
the reason for having $3\dmax\underline{+1}$ slots in the $\Dsc$ arrays; the +1 is for creating a temporary dummy tree of this type.}
Hence, without loss of generality, we now assume $S^+$ is the entire leaf set of
$\T$ and that there are no potential shortcuts w.r.t.~$\T$.

Notice that, after adding $(i', t')$-status to the leaves in $\T$, every $(i, t)$-branching node of depth at least $i'$ in $\T$
is also an $(i', t')$-branching node. Moreover, for each such $(i, t)$-branching node, adding $(i', t')$-status to the node
converts at most one $\H$-node into a new $(i', t')$-branching node.

Define ${\T}^*$ to be the subtree of $\H$ induced by \underline{all} ancestors of leaves in $\T$ up to depth $i$.  Our first
task is to enumerate all nodes of ${\T}^*$ at depth $i'$; call them $r_1,\ldots,r_k$.

\begin{claim}
The nodes $r_1,\ldots,r_k$ can be enumerated in worst case $O(k\log\log n)$ time.
\end{claim}

\begin{proof}
We perform a depth first search of $\T$ looking for nodes at depth $i'$.  Let $x$ be the locus of the search; initially $x$ is the root of $\T$.
If $x$ is at depth $i'$ we output $x$ and backtrack.  If $x$ is a $\T$-branching node we continue the search recursively on each $\T$-child
of $x$.  If $x$ has a single downward $\T$-shortcut $\scut{x}{x'}$ and $x'$ has depth \emph{strictly} greater than $i'$ we iteratively
uncover the downward shortcut from $x$ until it is $\scut{x}{x''}$, where $x''$ has depth at most $i'$, and move the locus of the search to $x''$.
If $k$ nodes are output by this procedure, the number of shortcuts followed/uncovered is $k\cdot O(\log\log n)$.
\end{proof}

Let $\T_1,\ldots,\T_k$ be the subtrees of $\T$ rooted at $r_1,\ldots,r_k$ and let $\W_1,\ldots,\W_k$ be the $(i',t')$-trees rooted at these nodes.
It may be that some $r_l$ does not currently have $(i',t')$-status, in which case $\W_l$ is empty.  In this case we simply traverse $\T_l$,
giving each node encountered $(i',t')$-status.  In Claim~\ref{lemma:merge-two-trees} we focus on the non-trivial problem of merging
$(\T_l,\W_l)$ when $r_l$ is an existing $(i',t')$-root.
Here ``$\W_l$-status'' is synonymous with $(i',t')$-status.

\begin{claim}\label{lemma:merge-two-trees}
Let $\T_l,\W_l$ be two trees rooted at $r_l$, where all shortcuts are maximal.
We can give $\W_l$-status to all leaves of $\T_l$ (and find all new $\W_l$-branching vertices)
in amortized $O(|\T_l|(\log\log n)^2)$ time, independent of the size of $\W_l$.
\end{claim}

\begin{figure}[ht]
\centering
\begin{minipage}[b]{0.32\linewidth}
\begin{subfigure}[t]{\linewidth}
\includegraphics[width=\linewidth]{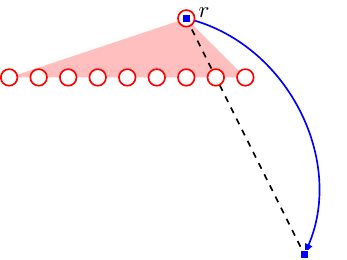}%
\caption{Case 1a}
\end{subfigure}\\
\begin{subfigure}[t]{\linewidth}
\includegraphics[width=\linewidth]{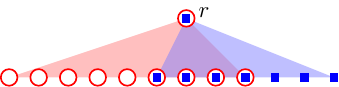}%
\caption{Case 1b}
\end{subfigure}%
\end{minipage}
\begin{subfigure}[t]{0.32\linewidth}
\includegraphics[width=\linewidth]{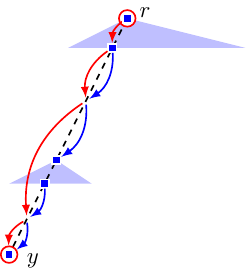}%
\caption{Case 2a}
\end{subfigure}%
\begin{subfigure}[t]{0.32\linewidth}
\includegraphics[width=\linewidth]{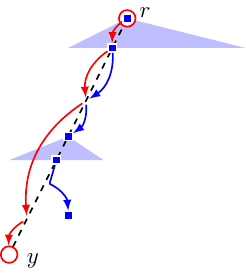}%
\caption{Case 2b}
\end{subfigure}
\caption{
    The examples to the four cases in the proof of Lemma~\ref{lemma:merge-two-trees}. The red circle nodes are $\T$-nodes and the blue square nodes are $\W$-nodes.
}\label{figure:lemma7.2-cases}
\end{figure}

\begin{proof}
We merge $\T_l$ and $\W_l$ in a depth-first manner.  Let $r$ be the locus of the search; initially $r=r_l$.  We maintain the invariant
that $r$ is both a $\T_l$-node and a $\W_l$-node.  There are two main cases; Case 1 is when $r$ is a branching $\T_l$-node
and Case 2 is when $r$ is a single-child $\T_l$-node.  
See Figure~\ref{figure:lemma7.2-cases} for illustration.

\paragraph{Case 1a:} \emph{$r$ is a branching $\T_l$-node but not a branching $\W_l$-node.}
After the merging process $r$ will become a branching $\W_l$-node, and therefore can have no downward $\W_l$-shortcut.
We repeatedly uncover the $\W_l$-shortcut leaving $r$.  In the final step
we uncover a fundamental shortcut $\scut{r}{x}$,
give 
$\ell_x$ \emph{local} $\W_l$-status in $\L(r)$,
and then designate $r$ a branching $\W_l$-node.
This reduces the situation to Case 1b.

\paragraph{Case 1b:} \emph{$r$ is both a branching $\T_l$-node and branching $\W_l$-node.}
Enumerate every $\T_l$-child $r'$ of $r$.  If $r'$ does not have $\W_l$-status, traverse the entire subtree of $\T_l$ rooted at $r'$,
marking each node encountered as a $\W_l$-node, 
and give $\ell_{r'}$ local $\W_l$-status in $\L(r)$.
Otherwise, move the locus of the search to $r'$ and recursively merge the
subtrees of $\T_l$ and $\W_l$ rooted at $r'$.

\paragraph{Case 2a:} \emph{$r$ is a single-child $\T_l$-node 
and the $\T_l$-child of $r$ is a $\W_l$-node or has a downward $\W_l$-shortcut.}
Let $y$ be the $\T_l$-child of $r$.  If $y$ is a $\W_l$-node then there are no \emph{new} branching vertices on the path from $r$ to $y$ (exclusive).
In this case we move the locus of the search to $y$ and continue recursively.   
If $y$ is not a $\W_l$-node but has a downward
$\W_l$-shortcut it becomes a branching $\W_l$-node.  
We repeatedly uncover its downward $\W_l$-shortcut, culminating
in uncovering a fundamental shortcut $\scut{y}{x}$, then designate $\ell_x$
a local $\W_l$-node in $\L(y)$ and designate $y$ a branching
$\W_l$-node.  Finally we move the locus of the search to $y$.

\paragraph{Case 2b:} \emph{$r$ is a single-child $\T_l$-node, but its $\T_l$-child $y$ is neither a $\W_l$-node nor
has a $\W_l$-shortcut.}
In this case, $y$ will become a branching $\W_l$-node or $\W_l$-leaf.  In addition, there may be a new branching $\W_l$-node on the path from $r$ to $y$.  We proceed to find the new branching node as follows.
Initialize $x=r$ and let $\scut{x}{x'}$ refer to its current downward $\T_l$-shortcut.
Whenever $x'$ is a $\W_l$-node or has a $\W_l$-shortcut, 
we move the locus of the search to $x'$, setting $x=x'$.
Whenever $x$ has a downward $\T_l$-shortcut $\scut{x}{x'}$
and a $\W_l$-shortcut $\scut{x}{x''}$ with $x'\neq x''$,
we uncover the one with maximum power, or uncover both
if they have the same power.  If $\scut{x}{x''}$ does not exist
because $x$ is a branching $\W_l$-node
then we repeatedly uncover $\scut{x}{x'}$.  
Eventually this process terminates when we uncover
a fundamental $\T_l$-shortcut $\scut{x}{x'}$
(perhaps uncovering a fundamental $\W_l$-shortcut $\scut{x}{x''}$
at the same time).
Then $x$ is the new branching $\W_l$-node.  We designate it as
such, and explore the $\T_l$ subtree rooted at $x'$,
giving all $\T_l$-nodes and shortcuts encountered $\W_l$-status.

\paragraph{About Invariant~\ref{invariant:it-shortcuts}.}
Notice that all new $(i', t')$ branching nodes are correctly identified by the procedure described above, and that $i'\ge i$. Thus, Invariant~\ref{invariant:it-shortcuts} holds for all $\H$-nodes of depth $\ge i$.

\paragraph{Time Complexity.}
The time required to traverse $\T_l$ and identify all new branching nodes is $O(|\T_l|\log\log n)$.
The running time is dominated by the cost of introducing up to $O(|\T_l|)$ new branching vertices and adding $\W_l$-status to $O(|\T_l|)$ nodes.
The cost of adding $\W_l$-status is $O((\log\log n)^2)$ and the cost of uncovering a fundamental $\W_l$-shortcut, in Case 1a or Case 2b,
is also $O((\log\log n)^2)$.  In total the time is $O(|\T_l|(\log\log n)^2)$.
\end{proof}

\subsection{Approximate Counters Operations --- Proof of Lemma~\ref{lemma-counters}}\label{section:proof-of-lemma-counters}

\paragraph{Update ancestor approximate $i$-counters.}
The data structure updates the approximate $i$-counters from a given $\H$-leaf $x$ to the corresponding $\H$-root.
Let $v$ be the current $(i, \const{primary})$-node.  If $v$ is a single-child $(i, \const{primary})$-node, then it adopts the approximate
$i$-counter of its $(i, \const{primary})$-child.  If $v$ is the child of an $(i, \const{primary})$-branching node $u$,
the data structure updates the approximate $i$-counters of $v$ from $\L(u)$ using Lemma \ref{lemma:local-tree-operations}.
At this point $u$ adopts
the approximate $i$-counter of the root of $\L(u)$.  There are at most $\log n$ branching nodes on the path and each costs $O((\log\log n)^2)$ time
to update an $i$-counter (Lemma \ref{lemma:local-tree-operations}), for a total of $O(\log n(\log\log n)^2)$ time.

\paragraph{Update approximate $i$-counters in an $(i, \const{primary})$-tree $\T$ rooted at $u^i$.}
At the beginning of this operation, the approximate $i$-counters at all $(i, \const{primary})$-leaves are accurate but
those at internal nodes are presumed invalid.  Beginning at the root $u^{i}$, the data structure
traverses the $(i, \const{primary})$ tree $\T$ in a postorder fashion, setting approximate $i$-counters in this order.
As in the analysis above, the cost is $O((\log\log n)^2)$ per node in $\T$,
for a total of $O(|\T|(\log\log n)^2)$.

\paragraph{Update approximate counters at a merged/split $\H$-node $x$.}
Suppose $x={\bf u^i}$ is the result of merging several siblings.  We inspect the
root of $\L(x)$ and retrieve the bitmap $I$ indicating for which $(i,\const{primary})$-pairs
$x$ is an $(i,\const{primary})$-branching node.  Using table lookups,
in $O(\log\log n)$ time we make
an $O(\log n\log\log n)$-bit mask and copy all the approximate $i$-counters from the root of $\L(x)$ to $x$.
The case when $x$ is the result of a split is handled in the same way.

\section{Amortized Analysis of Shortcut Maintenance}\label{section:amortized-analysis-of-shortcut-maintenance}

In this section, we describe how shortcuts are utilized and supported on $\H$.
Moreover, we provide a potential function
for $\H$-shortcuts that contributes to the amortized analysis for the \Delete{} operation.

\subsection{Covering All Shortcuts Touching Specified Paths --- Proof of Lemma~\ref{lemma:uncover-a-path}}%
\label{section:uncover-a-path}

The remainder of this section constitutes a
proof of Lemma~\ref{lemma:uncover-a-path}.
Let $P$ be a path from the given $\H$-node $u^i$ to the corresponding $\H$-root $u^0$.

\paragraph{Uncover and remove all $\H$-shortcuts touching $P$.}

Removing a fundamental shortcut is a local tree operation that costs 
$O((\log\log n)^2)$ time. 
Uncovering a shortcut with both endpoints on the path costs $O(\log\log n)$ time by Lemma~\ref{lemma:shortcuts-operation}.
(Such a shortcut may be an $(i',t')$-shortcut
for multiple $(i',t')$ pairs.)
Uncovering a non-fundamental deviating $(i, t)$-shortcut costs $O(1)$ time, 
by setting the appropriate $(i,t)$-bits in the supporting shortcuts.
Thus, the total cost of uncovering and removing all of the $\H$-shortcuts 
on the affected paths is $O(\log n (\log \log n)^2)$.

For each $\H$-node $x$ iterated from $u^0$ to $u^i$, the data structure first enumerates all downward $\H$-shortcuts in $\Dsc_x$. Then the data structure repeatedly uncovers the $\H$-shortcut with the largest power $>0$ until every $\H$-shortcut leaving $x$ is fundamental.

The data structure then uncovers each fundamental $\H$-shortcut leaving $x$ by the following procedure.
To uncover (remove) a fundamental $\H$-shortcut $\scut xy$, the data structure first detaches the local leaf $\ell_y$ in $\L(x)$ representing $y$ and re-inserts $\ell_y$ into the buffer tree.
Notice that this operation does not alter the structure of $\H$, so any $\H$-shortcut leaving $y$ is not affected.
Then the data structure adds local $(i, t)$-status to $\ell_y$ 
for all $(i, t)$ pairs indicated in the bitmap $b_{\scut xy}$.
This enables one to navigate from the root of $\L(x)$ to $\ell_y$ via \emph{local} $(i,t)$-shortcuts in $\L(x)$.  To preserve Invariant~\ref{invariant:it-shortcuts} (and thereby keep the whole $(i,t)$-forest in $\H$ navigable) we designate $x,y$ $(i,t)$-nodes for each $(i,t)$-bit indicated in $b_{\scut xy}$.
By the local tree operations listed in Lemma~\ref{lemma:buffer-tree-operations}, the time cost for uncovering (removing) a fundamental $\H$-shortcut is amortized $O((\log\log n)^2)$.

\paragraph{Adding a fundamental shortcut between an $\H$-node $v$ and its $\H$-parent $u$ for all $(i, t)$ pairs indicated by the bit vector $b$.}
This can be done by first invoking Lemma~\ref{lemma:local-tree-operations}, removing local $(i, t)$-status from $\ell_v$, and then adding a shortcut $\scut{u}{v}$ via Lemma~\ref{lemma:shortcuts-operation}. The time cost is $O(\log\log n)$.

\paragraph{Adding all fundamental $\H$-shortcuts touching $P$ shared by some $(i, t)$ pairs.}
There are two types of fundamental $\H$-shortcuts touching $P$: (1) having both endpoints on $P$, and (2) deviating from $P$.

To add all fundamental $\H$-shortcuts touching $P$, the data structure checks for edge depth $j$ iterated from $i$ to $1$ whether to add the fundamental shortcut $\scut {u^{j-1}}{u^j}$ or not.
It should be added if, for some $(i,t)$ pair, $u_j$ is an $(i,t)$-node but $u^{j-1}$ is \emph{not} an $(i,t)$-branching node.
To check this, the data structure first obtains a bitmap $b$ stored in $u^j$ indicating which $(i, t)$ pairs have an $(i, t)$-status at
$u^j$, and then accesses the path in the local tree $\L(u^{j-1})$ from $\ell_{u^{j}}$ to the root of $\L(u^{i-1})$.
During this traversal, if we encounter a \emph{local} 
$(i,t)$-branching node we set the corresponding 
$(i,t)$-bit in $b$ to zero.
When we reach the root of $\L(u^{j-1})$, if $b$ is still non-zero,
the data structure creates the fundamental $\H$-shortcut $\scut {u^{j-1}}{u^j}$ with $b_{\scut {u^{j-1}}{u^j}} = b$.
Furthermore, for each $(i, t)$-bit set to $1$ in $b$,
the data structure removes local $(i, t)$-status 
from the local tree leaf $\ell_{u^j}$.  
If $u_j$ is not an $(i,t)$-branching node, we also remove
$(i,t)$-status from $u_j$.

To handle the second case, notice that by Lemma~\ref{lemma:shortcuts-touching-a-path}, for each $(i, t)$ pair there is at most one fundamental $(i, t)$-shortcut deviating from $P$. In particular, for an $(i, t)$ pair, at most one deviating fundamental $(i, t)$-shortcut is added touching the unique $\H$-node $u^{j-1}$ 
such that $u^{j-1}$ belongs to an $(i, t)$-forest but $u^{j}$ 
does not.  
The data structure forms the bitmap $\mathit{diff}$ 
in $O(1)$ time indicating all such pairs. 
For each $(i,t)$ in $\mathit{diff}$ we check in $O(\log\log n)$
time whether $\L(u^{j-1})$ contains a \underline{single} leaf $\ell_y$ with local $(i,t)$-status (Lemma \ref{lemma:local-tree-operations}).  If so, we create a fundamental shortcut $\scut{u^{j-1}}{y}$, remove local $(i,t)$-status from $\ell_y$,
and remove $(i,t)$-status from $y$ if it is not an $(i,t)$-branching node. 

We now analyze the time cost.  For (1), at most $O(\log n)$ $\H$-shortcuts are covered, and each covering involves multiple $(i, t)$ pairs so each covering can be done in $O(H(u^{j})-H(u^{j+1}))$ time (Lemma \ref{lemma:local-tree-operations}), which telescopes to $O(\log n\log\log n)$.
Moreover, removing $(i, t)$-status on local tree leaves costs 
$O(\log\log n)$ time, by Lemma~\ref{lemma:local-tree-operations}.
For (2), there are $O(\log n)$ possible 
deviating fundamental shortcuts
to be created.  Each requires $O(\log\log n)$ amortized 
time, for a total of $O(\log n\log\log n)$ amortized time.

\paragraph{Cover all $(i, t)$-shortcuts having both endpoints on $P$.}
In addition to adding all of the fundamental shortcuts, the data structure adds back all of the $\H$-shortcuts on the path $P$
from $u^j$ to $u^0$.
This is done by traversing $P$ $\log\log n$ times.
In the $p$-th traversal the data structure covers all possible $\H$-shortcuts of power $p+1$ that have both endpoints on the path. Each shortcut is covered in $O(\log\log n)$ time: to cover $\scut xy$ from $\scut x{y'}$ and $\scut {y'}y$, the data structure first adds the shortcut $\scut xy$ into $\Usc_y$. Then the data structure computes the bitwise AND of two bitmaps by setting $b_{\scut xy} \gets b_{\scut x{y'}}\land b_{\scut {y'}y}$, and removes the bits in the covered shortcuts by setting $b_{\scut x{y'}} \gets b_{\scut x{y'}} \oplus b_{\scut xy}$ and $b_{\scut {y'}y}\gets b_{\scut {y'}y} \oplus b_{\scut xy}$.
Finally, the data structure updates 
$\Uptr_y$, $\Dsc_x$ and $\Dptr_x$ according to $b_{\scut xy}$,
and if $b_{\scut x{y'}}$ and/or $b_{\scut {y'}y}$ becomes 0, 
updates $\Dsc_x,\Dptr_x,\Occ_x,\Dsc_{y'},\Dptr_{y'},\Occ_{y'}$ appropriately.

It is straightforward to see that, after $\log\log n$ passes, 
if there is any $(i, t)$-shortcut with at least one endpoint on the path that could be covered, the other endpoint must be outside of the 
path and hence is a deviating $(i, t)$-shortcut.
Since there are a total of $O(\dmax)=O(\log n)$ non-fundamental $\H$-shortcuts to consider, the total time cost is 
$O(\log n \log\log n)$.

\subsection{Shortcut Cost Analysis}\label{section:cost-analysis-lazy-covering}

At first glance it seems sensible to charge the
cost of deleting a shortcut to the creation of the shortcut,
and therefore only account for their creation in the amortized
analysis.  This does not quite work because shortcuts are
\emph{shared} between many $(i,t)$ pairs and the cost of deleting
a shortcut depends on how broadly it is shared.
The amortized analysis for $\H$-shortcuts focusses on \emph{supporting potential shortcuts} defined as follows:

\begin{definition}\label{definition:lazy-cover-potential}
Let $u$ be a single-child
$(i, t)$-node and $v$ be the $(i, t)$-child of $u$. Then the \emph{maximal potential $(i, t)$-shortcuts} are the
maximal shortcuts with respect to the covering
relation having both endpoints on the path $P_{uv}$.
The \emph{supporting potential $(i,t)$-shortcuts} are the
$\H$-shortcuts that support some maximal potential $(i,t)$-shortcut.
\end{definition}

Consider a supporting potential shortcut
$\scut uv$ (which may or may not be \emph{stored})
and define $k_{\scut uv}$ to be the number of $(i,t)$
pairs for which $\scut uv$ is covered by
a maximal potential $(i,t)$-shortcut but \emph{is not}
covered by a \emph{stored} $(i,t)$-shortcut.\footnote{
The count $k_{\scut uv}$ also takes the dummy tree
into account, as if it had a special $(i,t)$-status.
Notice that the dummy tree only exists in the middle of
the $\Delete$ operation; see Section~\ref{section:proof-lemma-induced-forests}.}  Define a function $f$ as follows.
\begin{align*}
f(\scut uv) &= \begin{cases}
k_{{\scut uv}}, 	&	 \text{if $\scut uv$ is not a fundamental shortcut,}\\
0, 	&	 \text{if $\scut uv$ is a fundamental shortcut.}
\end{cases}
\end{align*}

Let $C$ be the set of all shortcuts defined over $\H$,  $C_{\operatorname{st}}$ be the set of all stored non-fundamental shortcuts,
and $C_{\operatorname{f}}$ be the set of all stored fundamental shortcuts. The potential $\Phi$ is defined as follows.
\[
 \Phi = \underbrace{\left(\sum_{\scut uv \,\in\, C} f(\scut uv)(\log\log n+1)\right)}_{\Phi_1}
 \;+\;
 \underbrace{|C_{\operatorname{st}}|\cdot \log\log n}_{\Phi_2}
 \;+\;
 \underbrace{|C_{\operatorname{f}}|\cdot (\log\log n)^2}_{\Phi_3}
\]
 Uncovering a fundamental shortcut could possibly cause a detach-reattach operation in the local tree, which costs $O((\log\log n)^2)$ time; see the proof of Lemma~\ref{lemma:uncover-a-path} in Section~\ref{section:uncover-a-path}. 
 This is the reason that we give more credit to a stored fundamental shortcut than to a non-fundamental shortcut.
 Throughout the algorithm execution, there are many places where the $(i, t)$-forests are modified. These structural changes affect the potential $\Phi$ so we list them in the following paragraphs. 

\paragraph{Adding $(i, t)$-status to an $\H$-leaf. (Lemma~\ref{lemma-induced-forests})}
Adding $(i, t)$-status to an $\H$-leaf increases $\Phi$ by $O(\log n(\log\log n)^2)$ since
all new shortcuts that need to be created lie on the path from the
leaf to its $(i, t)$-parent.  In particular, each of the $O(\log n)$ new fundamental shortcuts increases $\Phi_3$ by $(\log\log n)^2$ each,
and both $\Phi_1$ and $\Phi_2$ increase by at most
$O(\log n\log\log n)$ each.

\paragraph{Removing $(i, t)$-status from an $\H$-leaf. (Lemma~\ref{lemma-induced-forests})}
Removing $(i, t)$-status from a leaf $x$ increases $\Phi$ by $O((\log\log n)^2)$.
Let $y$ be the $(i,t)$-parent of $x$.
If $y$ loses its $(i,t)$-status and its $\H$-parent $z$ is no longer
an $(i,t)$-branching node, we will create one new
fundamental shortcut from $z$ to a sibling of $y$, 
increasing $\Phi_3$ by $(\log\log n)^2$.
All new supporting potential $(i,t)$-shortcuts will cover $z$ and
have distinct powers.
Thus, the net increase of $\Phi_1$ will be at most
$(\log\log n+1)\log\log n$.  $\Phi_2$ is unchanged.

\paragraph{Creating a dummy tree. (Lemma~\ref{lemma-induced-forests})}
Create a dummy tree $\T$ by copying a maximally covered
$(i, t)$-tree.  Recall that there are $3\dmax+1$ shortcut forests,
one for every $(i,t)$-pair and 1 for the dummy forest; we will say its shortcuts have \emph{$\perp$-status}.
After creating the dummy tree $\T$ and giving its maximal shortcuts
$\perp$-status, there is no change to $\Phi$.  Every potential
$\perp$-shortcut is a stored shortcut, and was formerly stored before $\T$ was created.

\paragraph{Removing $(i, t)$-status from a subset of $\H$-leaves. (Lemma~\ref{lemma-induced-forests})}
The data structure removes $(i, t)$-status (or $\perp$-status) from a subset of leaves in an $(i, t)$-tree $\T$ (or dummy tree $\T$).
There are $O(|\T|)$ leaves removed, and each removal
increases $\Phi$ by at most $O((\log\log n)^2)$, for a total
of $O(|\T|(\log\log n)^2)$.

\paragraph{Merging and destroying dummy trees. (Lemma~\ref{lemma-induced-forests})}
The data structure 
merges a maximally covered dummy tree $\mathcal{T}$
into an $(i', t')$-tree, and destroys $\mathcal{T}$.
Observe that in the process of merging these trees, the $(i',t')$-tree
acquires new branching nodes and the set of supporting potential $(i',t')$-shortcuts only loses elements.  Thus $\Phi_1$ does not increase.  Every shortcut supporting the merged tree was in at least
one of the two original trees before the operation, so $\Phi_2$ and $\Phi_3$ are also non-increasing.

\paragraph{Lazy Covering. (Lemma~\ref{lemma:lazy-cover-operation})}
The lazy covering method only covers non-fundamental shortcuts, so each covering costs constant actual time.
Suppose we have traversed $(i,t)$-shortcuts
$\scut xy$ and $\scut yz$ and covered
them with $\scut xz$.  (Notice that $\scut xz$ may or may not have been previously \emph{stored}.)
This causes $f(\scut xz)$ to drop by at least 1 and hence $\Phi_1$ to drop by $\log\log n+1$.  If $\scut xz$ was not already stored, $\Phi_2$ increases by $\log\log n$.  In any case, the net potential drop in $\Phi$ is at least 1, which pays for the covering.

\paragraph{The $\Delete$ Operation. (See also Section~\ref{subsubsection:detailed-deletion})}
At the beginning of a $\Delete(u,v)$ operation, the algorithm spends $O(\log\log n)$ time to
uncover each $\H$-shortcut touching an ancestor of
$u^i$ or $v^i$, where $i$ is the depth of $\{u,v\}$.
Notice that these $\H$-shortcuts may be shared by
many $(i', t')$-pairs,
so the uncovering operation may \emph{temporarily}
increase $\Phi_1$ by $\Omega(\log^2 n\log\log n)$.
Fortunately, after the deletion operation most of these $\H$-shortcuts are covered back.
As mentioned in Section~\ref{section:uncovering-a-path},
after a deletion
the data structure covers every possible supporting potential
$(i', t')$-shortcut with both endpoints at ancestors of $u^i$ or $v^i$, as well as all necessary fundamental $\H$-shortcuts
with at least one endpoint ancestral to $u^i$ or $v^i$.
We claim that after covering back all necessary $\H$-shortcuts
on the two paths, the increase of $\Phi$
is upper bounded by $O(\log n(\log\log n)^2)$.
Counting multiplicity, there are $O(\log n\log\log n)$ non-fundamental deviating shortcuts that the lazy covering method failed to restore after the $\Delete$ operation.
Each contributes $\log\log n+1$ to $\Phi_1$, for a total of $O(\log n(\log\log n)^2)$.
The number of non-fundamental shortcuts with both endpoints
at ancestors of $u^i$ or $v^i$ is $O(\log n)$, and each contributes $\log\log n$ to $\Phi_2$, for a total of $O(\log n\log\log n)$.  Similarly, the $O(\log n)$ fundamental shortcuts each contribute $(\log\log n)^2$ to $\Phi_3$, for a total of $O(\log n(\log\log n)^2)$.
The increase in $\Phi$ due to these changes are charged to the $\Delete$ operation.





\section{Main Operations --- Proof of Lemma~\ref{lemma-main}}\label{section:main-oerations}

We review how each of the 10 operations of
Lemma~\ref{lemma-main} can be implemented in the stated amortized running time.

\paragraph{Operation \ref{opr:add-1-witness} --- 
Add or remove an edge with depth $i$ and endpoint type $t$.}
The data structure first adds (or removes) the given edge to the $\H$-leaf data structures of its endpoints; see Section~\ref{section:overview:leaf-data-structure}.
If the addition/removal changes the $(i,t)$-status of either endpoint, we update them with
Lemma~{\ref{lemma-induced-forests}}
and if $t=\const{primary}$
we update the approximate $i$-counters
using Lemma~\ref{lemma-counters}.
The time cost is $O(\log n(\log\log n)^2)$.

\paragraph{Operation \ref{opr:merge-two-siblings} --- Merge a subset of  $\H$-siblings into $\bf u^i$ and promote all $i$-witness edges touching $\bf u^i$.}%

Given the subset $S$ of $\H$-siblings at depth $i$,
the algorithm first uncovers all $\H$-shortcuts that
touch any $\H$-siblings in $S$ (Lemma~\ref{lemma:lazy-cover-operation}).
We then invoke Lemma~\ref{lemma:local-tree-operations}
to merge $\H$-siblings in $S$, two at a time, into a single $\H$-node $\bf u^i$.
The amortized cost for uncovering and deleting all $\H$-shortcuts
touching $S$ is zero.
(The cost for \emph{restoring} necessary
shortcuts is not part of this operation.
It is paid for by the $\Delete$ itself; see Section~\ref{section:cost-analysis-lazy-covering}.)
Thus, by Lemma~\ref{lemma:local-tree-operations},
the amortized cost so far is $O(|S| (\log\log n)^2)$.

The algorithm then traverses the $(i, \const{witness})$-tree
rooted at $\bf u^i$,
obtains the set of leaf-descendants with $(i, \const{witness})$-status and enumerates the $|S|-1$
$(i, \const{witness})$-edges touching these vertices.
By Lemmas~\ref{lemma:lazy-cover-operation} and \ref{lemma:local-tree-operations},
the amortized cost of the traversal is $O(|S|\log\log n)$.
Now the data structure uses Lemma~{\ref{lemma-induced-forests}} (last bullet point) to promote
all these $(i, \const{witness})$-edges to $(i+1, \const{witness})$-status, which
costs $O((|S|-1)(\log\log n)^2)$ time.

Notice that every edge releases 
{$\Omega((\log\log n)^2)$} units of potential upon
promotion. 
As every unit of potential pays for some constant $\Theta(1)$ running time, 
the amortized cost of this operation can be made
$-\Omega((|S|-1)(\log\log n)^2)$
by choosing a sufficiently large constant.

\paragraph{Operation \ref{opr:convert-i-secondary} --- Upgrade all $i$-secondary endpoints touching $\bf u^i$.}
The data structure first traverses the $(i, \const{secondary})$-tree rooted at $\bf u^i$, enumerating its leaf-set $S$.
By Lemma~\ref{lemma-induced-forests},
enumerating $S$ costs $O(|S|\log\log n)$ time.
Let $s\ge |S|$ be the number of $(i, \const{secondary})$-endpoints 
stored at these leaves.
We then use Lemma~{\ref{lemma-induced-forests}}
to add $(i, \const{primary})$-status and remove $(i, \const{secondary})$-status
from all leaves in $S$, in $O(|S|(\log\log n)^2)$ amortized time.
Using the $\H$-leaf data structure, we can upgrade all $s$ $(i,\const{secondary})$-endpoints
to $(i,\const{primary})$-status in $O(s)$ time.
At this point the approximate $i$-counters at $S$ are accurate, but the approximate $i$-counters at ancestors of $S$ are out of date.
Using Lemma~\ref{lemma-counters}, we rebuild all approximate $i$-counters at descendants of $\bf u^i$ in
$O(p(\log\log n)^2)$ time, where $p\ge |S|$ is the number of
$(i, \const{primary})$-leaves descending from $\bf u^i$.

The $s$ upgrades release {$\Omega(s(\log\log n)^2)$} units of potential
whereas the cost for traversing the $(i,\const{primary})$-tree and updating
its counters is $O(p(\log\log n)^2)$.  Thus, the amortized time
of this operation is $-\Omega((s-p)(\log\log n)^2)$.

\paragraph{Operation \ref{opr:promote-i-primary} --- Promote a subset of $i$-primary endpoints touching $\bf u^i$.}
Let $R$ be the set of $(i,\const{primary})$ endpoints being promoted.
The data structure first scans through $R$, forming two leaf sets:
$S^-$ are all $\H$-leaves whose $(i,\const{primary})$-endpoints
are contained in $R$ (these will lose $(i,\const{primary})$-status)
and $S^+$ are all $\H$-leaves touched by at least one element of $R$
(these will gain $(i+1,\const{secondary})$-status, if they do not have it already).
Both $S^-$ and $S^+$ are leaves of the
$(i, \const{primary})$-tree $\T$ rooted at $\bf u^i$.
The data structure
uses Lemma~{\ref{lemma-induced-forests}} to add
$(i+1, \const{secondary})$-status to all $\H$-leaves in $S^+$ and
removes $(i, \const{primary})$-status from all $\H$-leaves in $S^-$.
By Lemma~{\ref{lemma-induced-forests}} the time cost is
$O(|\T|(\log\log n)^2+1)$. Let $p$ be the number of $i$-primary endpoints touching
$\bf u^i$, including the ones that are not promoted.
Since $|\T|\le p$ we have that this operation costs $O(p(\log\log n)^2+1)$ time.

Since the promotions release
{$|R|\cdot \Omega((\log\log n)^2)$} units of potential,
with the leading constants set properly
the amortized cost of this operation is
at most $-\Omega((12|R|-p)(\log\log n)^2)$.

\paragraph{Operation \ref{opr:convert-i-non-witness} --- Convert an $i$-non-witness edge to an $i$-witness edge.}

The data structure changes the status of the  endpoints of the converted edge to $(i,\const{witness})$ using the $\H$-leaf data structure.
If either endpoint of the edge had $(i,\const{primary})$-status prior to the conversion, the approximate $i$-counters at all ancestors of the $\H$-leaf  containing the endpoing
may be invalid and the endpoints may lose $(i,\const{primary})$-status.
The data structure updates the approximate $i$-counters at all $(i, \const{primary})$-ancestors,
and removes $(i,\const{primary})$-status of the endpoints, if necessary.
This costs $O(\log n(\log\log n)^2)$ time, by
Lemmas~\ref{lemma-counters} and \ref{lemma-induced-forests}.

\paragraph{Operation \ref{opr:split-node} --- Split an $\H$-node $u^{i-1}$ with a single child $u^i$.}
We are given pointers to $u^{i-2}$ (if it exists), $u^{i-1}$, and $u^i$.
The data structure first creates a new $\H$-node $x$,
detaches $u^i$ from $\L(u^{i-1})$, and makes $u^i$ a child of $x$ using Lemma~\ref{lemma:local-tree-operations}.
If $i=1$, then $x$ is an $\H$-root and we are done.
Otherwise, the data structure attaches $x$ to $\L(u^{i-2})$.   By Lemma~\ref{section:local-tree-cost-analysis}, the amortized time for all these operations is $O((\log\log n)^2)$.

\paragraph{Operation \ref{opr:all-i-witness} --- Enumerate all $(i, t)$-endpoints in the $(i, t)$-tree rooted at $\bf u^i$.}
The data structure traverses the $(i, t)$-tree.  For each $(i, t)$-leaf, enumerate all the endpoints of depth $i$ and type $t$
from the $\H$-leaf data structure.
By applying the operations of  Lemma~\ref{lemma-induced-forests},
the time cost is $O(l\log\log n + k) = O(k\log\log n)$, where $l$ is the number of $(i,t)$-leaves and $k$
is the number of enumerated endpoints.

\paragraph{Operation \ref{opr:find-i-component} --- Accessing $\H$-parent $v^{i-1}$ from $v^i$.}
This is a local tree operation.  According to Lemma~\ref{lemma:local-tree-operations}, the time cost is $O(H(v^{i-1}) - H(v^i))$.

\paragraph{Operation \ref{opr:counter} --- Accessing an approximate $i$-counter.}
The approximate $i$-counter is stored at the node in floating-point representation.  It can be retrieved and converted to an integer 
(Lemma~\ref{lemma:approximate-counters-addition}) in $O(1)$ time.  

\paragraph{Operation \ref{opr:sample} --- Batch Sampling Test.}
From Section~\ref{section:the-batch-sampling-test}, the batch sampling test on $k$ samples costs
worst case time
$O(\min((p+s)\log\log n+k, k\log n\log\log n))$ where $p$ is the number of $i$-primary edges touching $\bf u^i$ and $s$ is the number of $i$-secondary edges touching $\bf u^i$.

\subsection{Proof of Theorem~\ref{main-result}}\label{section:proof-of-main-result}


The correctness of the data structure follows from Section~\ref{section:deletion}'s
maintenance of Invariant~\ref{invariant:spanning-forest},
using Lemma~\ref{lemma-main} to maintain $\H$
and Theorem~\ref{lemma:witness-forest} to maintain the witness forest $\F$.
In this section we prove that the amortized time complexity is $O(\log n(\log\log n)^2)$ per $\Insert$ or $\Delete$ and $O(\log n/\log\log\log n)$ per
$\isConnected$ query.
Call $\DH$ the data structure for $\H$
described in Lemma~\ref{lemma-main} and $\DF$
the data structure for $\F$ from
Theorem~\ref{lemma:witness-forest}, fixing $t(n) = (\log\log n)^2$.

\subsubsection{Insertion} To execute $\Insert(u, v)$,
the algorithm makes a connectivity query to $\DF$ in
$O(\log n/\log t(n)) = O(\log n/\log\log\log n)$ time.
Then, there are two cases:
\begin{itemize}
   \item If $u$ and $v$ are already connected, then the algorithm invokes Operation~\ref{opr:add-1-non-witness} of Lemma~\ref{lemma-main} on the data structure $\DH$, adding the edge $\{u, v\}$ with depth $1$ and endpoint type $\const{secondary}$ in amortized $O(\log n(\log\log n)^2)$ time.
   \item Otherwise, $u$ and $v$ are not connected.
   The algorithm then invokes Operation~\ref{opr:find-i-component} $2\dmax$ times, obtaining pointers to $u^0$ and $v^0$.
   Thus, the cost of Operation~\ref{opr:find-i-component} telescopes to
   $O(\log n\log\log n)$ time.
   The algorithm then merges $u^0$ and $v^0$ using Operation~\ref{opr:merge-two-roots} in amortized $O((\log\log n)^2)$ time. Finally, $\{u,v\}$ is added to the data structure $\DH$ through Operation~\ref{opr:add-1-witness} as an edge with depth $1$ and type $\const{witness}$, in amortized $O(\log n(\log\log n)^2)$ time.
   The algorithm also inserts $\{u, v\}$ into $\DF$,
   in $O(\log n\cdot t(n)) = O(\log n(\log\log n)^2)$ time.
\end{itemize}

Hence, an $\Insert(u, v)$ operation costs amortized $O(\log n(\log\log n)^2)$ time.

\subsubsection{Deletion}\label{subsubsection:detailed-deletion} To execute a $\Delete(u, v)$ operation, where $e=\{u, v\}$,
the algorithm first removes $e$ from $\H$ through Operation~\ref{opr:remove-i-non-witness}, taking amortized
$O(\log n(\log\log n)^2)$ time. If $e$ is a non-witness edge, then the operation is done. Otherwise, the algorithm also removes $e$ from $\DF$ in $O(\log n\cdot t(n))$ time.
Then, the algorithm attempts to find a replacement edge iteratively at
depth $i=d_e, d_e-1, \ldots, 1$.

\paragraph{Preparing Iterations.}
As mentioned in Section~\ref{section:uncovering-a-path}, before the iterations begin, all ancestors of $u^{i-1}=v^{i-1}$ are found and stored in a list, using Operation~\ref{opr:find-i-component}.
The cost of Operation~\ref{opr:find-i-component} telescopes to
$O(\log n\log\log n)$ time.
In addition, all stored $\H$-shortcuts touching 
the path from $u^{i-1}$ to $u^0$ are uncovered, using Lemma~\ref{lemma:uncover-a-path}. 
We note that Invariant~\ref{invariant:it-shortcuts} now holds only for all $\H$-nodes at depth $\ge i$, which validates all operations whose implementation depends on Lemma~\ref{lemma-induced-forests}.
Once the shortcuts have been removed, the iterations begin.

\paragraph{Establishing Two Components.}
On the iteration concerning depth $i$, the algorithm runs two parallel
searches starting from $u^i$ and $v^i$,
obtaining the connected components $c_u$ and $c_v$.
Throughout the search, $\H$-siblings of $u^i$ and $v^i$ are found via $i$-witness edges enumerated by Operation~\ref{opr:all-i-witness}.
Let $S_u$ be the set of $\H$-siblings in the same component $c_u$ with $u^i$ and $S_v$ be the of $\H$-siblings for $c_v$ with $v^i$. Notice that there are exactly $|S_u|-1$ and $|S_v|-1$ $i$-witness edges in $c_u$ and $c_v$ respectively,
and each $i$-witness edge contributes 2 endpoints throughout the search.
Thus, the
searches in parallel take amortized
$O(\min\{|S_u|-1, |S_v|-1\}(\log\log n)+1)$ time
until the first completes.  At this point we can deduce which of $c_u$ or $c_v$ is the smaller weight component; suppose it is $c_u$.

The algorithm uncovers and removes 
all remaining downward shortcuts on
the siblings of $u^i$ that form $c_u$ (Lemma~\ref{lemma:uncover-a-path}), then
performs Operation~\ref{opr:promote-i-witness} to promote
all $(i,\const{witness})$-edges in $c_u$ to 
$(i+1,\const{witness})$ edges,
with a negative amortized cost of $-\Omega((|S_u|-1)(\log\log n)^2)$, 
which pays for the cost of the two searches.

In conclusion, establishing two components costs amortized constant time.

\paragraph{Finding a Replacement Edge.}
Recall from Section~\ref{section:the-batch-sampling-test}
that $\rho$ is the fraction of $i$-primary endpoints belonging to 
replacement edges and $p$ and $s$ are the number of primary and secondary endpoints.
When $\rho > 3/4$ the search for a replacement edge halts after the first or second batch sampling
test with probability $1-1/p$, 
and costs $O(\log n(\log\log n)^2)$ in expectation, which is charged to the $\Delete$ operation.
Suppose that the enumeration procedure is invoked, which 
upgrades \emph{all} $(i, \const{secondary})$ endpoints 
to $(i, \const{primary})$ status (Operation~\ref{opr:convert-i-secondary}), 
and then \emph{some} of the $(i, \const{primary})$ endpoints
to $(i+1, \const{secondary})$ status (Operation~\ref{opr:promote-i-primary}).
This procedure costs $O((p+s)\log\log n)$ time.
The amortized time cost of Operation~\ref{opr:convert-i-secondary} is $-\Omega((s-p)(\log\log n)^2)$.
At this point there are now $p'=p+s$ $(i,\const{primary})$ endpoints.
Suppose that Operation~\ref{opr:promote-i-primary} promotes $s'$ of them
to $(i+1,\const{secondary})$ status, at an amortized time cost of
$-\Omega((12s'-p')(\log\log n)^2) = -\Omega((12(1-\rho)p - (p+s))(\log\log n)^2)$.
(If $s'<p'$, then all the unpromoted endpoints belong to replacement edges.)
Let the leading constants of the amortized costs of Operations~\ref{opr:convert-i-secondary} and \ref{opr:promote-i-primary} be $c_0$ and $c_1$ times that of the cost of the enumeration procedure.
Then the amortized time cost of the enumeration procedure is proportional to
\begin{align*}
\lefteqn{\Big(\log\log n\Big)^2\Big(\big(p+s\big) - c_0\big(s-p\big) - c_1\big(12(1-\rho)p - (p+s)\big)\Big)}\\
&= \Big(\log\log n\Big)^2\Big(p\big(1 + c_0 - c_1(12(1-\rho) - 1)\big)
+
s\big(1 - c_0 + c_1\big)
\Big)
\end{align*}
When $\rho < 3/4$, the contribution of original primary 
endpoints ($p$) is at most $p(1+c_0-2c_1)(\log\log n)^2$,
which is at most 0 when $c_1\geq (1+c_0)/2$.  When $\rho>3/4$ the
enumeration procedure is invoked  with probability at most $1/p$, and the 
expected time cost is $O((\log\log n)^2)$.
Regardless of $\rho$, the contribution of original secondary endpoints
is $s(1-c_0+c_1)(\log\log n)^2$, which is at most 0 when $c_0 \geq c_1+1$.
Setting $c_0 = 3$ and $c_1=2$ satisfies both constraints.

In conclusion, \emph{successfully} 
finding a replacement edge in the first or second batch sampling test costs
$O(\log n(\log\log n)^2)$ expected time, which is charged to the $\Delete$ operation.
If the enumeration procedure is invoked, then the search for a replacement edge may fail
to find a replacement edge at level $i$.  
The amortized expected cost of the enumeration procedure at depth $i$ 
is $O((\log\log n)^2)$, which is charged to the $\Delete$ operation.

\paragraph{Preparation for Next Iteration.}
If no replacement edge is found at the current depth $i$, the algorithm splits $u^{i-1}$ into to $\H$-siblings $\bf u^{i-1}$ and $\bf v^{i-1}$, through Operation~\ref{opr:split-node}.
The split operation costs amortized $O((\log\log n)^2)$ time.
After the split, the algorithm restores all necessary downward shortcuts
touching $\bf u^{i-1}$, $\bf v^{i-1}$, $v^i$, or $\bf u^{i}$, as described in Section~\ref{section:uncovering-a-path} and Lemma~\ref{lemma:uncover-a-path}.
The covering of fundamental shortcuts ensures Invariant~\ref{invariant:it-shortcuts} to hold for all $\H$-nodes at depth $\ge i-1$.
By the same argument from Lemma~\ref{lemma:uncover-a-path},
the total cost of covering these shortcuts is
 $O(\log n (\log\log n)^2)$, which is charged to the $\Delete$ operation.

\paragraph{The End of Iteration.}
Suppose we find a replacement edge at depth $i$.  The algorithm ends
by restoring all necessary shortcuts with one endpoint at an ancestor
of $u^i$ or $v^i$.  
By Lemma~\ref{lemma:uncover-a-path}, this costs
$O(\log n(\log\log n)^2)$ time. 
Furthermore, this restores Invariant~\ref{invariant:it-shortcuts} holds for all nodes in $\H$.

\paragraph{Combining the Costs.}
Summing all costs, the total amortized 
expected time for
an edge deletion is $O(\log n(\log\log n)^2)$.

\section{Conclusion}\label{section:conclusion}

We have shown that the Las Vegas randomized amortized update time of dynamic connectivity is $O(\log n(\log\log n)^2)$, which leaves
a small $(\log\log n)^2$ gap between the cell probe lower bounds of \Patrascu{} and Demaine~\cite{PatrascuD06}
and \Patrascu{} and Thorup~\cite{PatrascuT11}.  
The main bottleneck in our approach is dealing with insertions in the
buffer trees inside local trees.  Each affects $O(\log\log n)$ local tree nodes, and the cost of updating such nodes involves adding $O(\log n)$ (floating point) approximate counters packed into $O(\log\log n)$ machine words.  If this $(\log\log n)^2$ barrier were overcome, there would 
still be a $\log\log n$-factor bottleneck, which arises from the shortcut infrastructure
and the height of the bottom, buffer, and top trees.  

It may be possible to achieve $O(\log n)$ amortized time in the Monte Carlo model with a private connectivity witness, by using connectivity sketches~\cite{AhnGM12,KapronKM13,GibbKKT15,Wang15,NelsonY19}.


\printbibliography

\end{document}